\documentclass[sn-basic]{sn-jnl}
\usepackage{graphicx}%
\usepackage{multirow}%
\usepackage{amsmath,amssymb,amsfonts}%
\usepackage{amsthm}%
\usepackage{mathrsfs}%
\usepackage[title]{appendix}%
\usepackage{xcolor}%
\usepackage{textcomp}%
\usepackage{manyfoot}%
\usepackage{algorithm}%
\usepackage{algorithmicx}%
\usepackage{algpseudocode}%
\usepackage{listings}%

\usepackage[utf8]{inputenc} 
\usepackage[T1]{fontenc}    
\usepackage{nicefrac}       
\usepackage{microtype}      
\usepackage{lipsum}
\usepackage{graphicx}       
\usepackage{todonotes}
\usepackage{natbib}  
\usepackage{caption} 
\usepackage{algorithm}

\usepackage{newfloat}
\usepackage{listings}
\DeclareCaptionStyle{ruled}{labelfont=normalfont,labelsep=colon,strut=off} 
\floatstyle{ruled}
\newfloat{listing}{tb}{lst}{}
\floatname{listing}{Listing}

\usepackage[utf8]{inputenc} 
\usepackage{lineno}
\nolinenumbers

\usepackage{url}            
\usepackage{booktabs}       
\usepackage{amsfonts}       
\usepackage{nicefrac}       
\usepackage{microtype}      
\usepackage{xcolor}         

\usepackage{amsmath, amssymb, amsthm}
\usepackage{mathtools}
\usepackage{verbatim}
\usepackage{colortbl}
\usepackage{thmtools}
\usepackage{thm-restate}

\usepackage{tikz}
\usepackage{pgfplots}
\pgfplotsset{compat=1.18}
\usetikzlibrary{math, calc}

\usepackage{booktabs}
\definecolor{DarkGreen}{rgb}{0.1,0.5,0.1}
\definecolor{DarkRed}{rgb}{0.5,0.1,0.1}
\definecolor{DarkBlue}{rgb}{0.1,0.1,0.7}
\usepackage{xspace}
\usepackage[mathscr]{euscript}
 \let\mathscr\relax
\usepackage[scr]{rsfso}

\usepackage[capitalise,nameinlink]{cleveref}

\usepackage{csquotes}

\usepackage{booktabs}
\usepackage{multirow}

\DeclareRobustCommand{\svdots}{
  \vbox{%
    \baselineskip=0.33333\normalbaselineskip
    \lineskiplimit=0pt
    \hbox{.}\hbox{.}\hbox{.}%
    \kern-0.2\baselineskip
  }%
}

\newcommand{\problem}{\textsc{proxy selection}}
\newcommand{\win}{\textnormal{\texttt{win}}}
\newcommand{\opt}{\textnormal{\texttt{opt}}}
\newcommand{\dRepallzero}{\textnormal{\texttt{dRep$\bar{0}$}}}
\newcommand{\dRepzero}{\textnormal{\texttt{dRep0}}}
\newcommand{\dRepone}{\textnormal{\texttt{dRep1}}}

\newcommand{\vecc}[1]{\ensuremath{\mathbf{#1}}}

\newcommand{\I}{C}
\renewcommand{\tilde}{\widetilde}

\def\epsilon{\varepsilon}
\renewcommand{\hat}{\widehat}
\renewcommand{\bar}{\overline}

\DeclareMathOperator*{\argmax}{\mathrm{arg\,max}}
\DeclarePairedDelimiter\ceil{\lceil}{\rceil}
\DeclarePairedDelimiter\floor{\lfloor}{\rfloor}

\usepackage{enumitem}

\input{insbox}

\newtheoremstyle{plain2}
{6pt}
{6pt}
{}
{}
{\bfseries}
{.}
{.5em}
{}

\theoremstyle{plain2}
\newtheorem{theorem}{Theorem}
\newtheorem{proposition}[theorem]{Proposition}
\newtheorem{lemma}{Lemma}

\newtheorem{claim}{Claim}

\newtheorem{example}{Example}%

\newtheorem{definition}{Definition}%

\begin{document}

\title[On the Potential and Limitations of Proxy Voting:\\
Delegation with Incomplete Votes]{On the Potential and Limitations of Proxy Voting:\\
Delegation with Incomplete Votes}

\author[1,5]{Georgios Amanatidis}
\author[2]{Aris Filos-Ratsikas}
\author[3]{Philip Lazos}
\author[3,4,5]{\\Evangelos Markakis}
\author[4,6]{Georgios Papasotiropoulos}

\affil[1]{University of Essex}
\affil[2]{University of Edinburgh}
\affil[3]{Input Output (IOG)}
\affil[4]{Athens University of Economics and Business}
\affil[5]{Archimedes\,/\,Athena Research Center}
\affil[6]{University of Warsaw}

\abstract{We study elections where 
voters are faced with the challenge of expressing preferences over an extreme number of issues under consideration. 
This is largely motivated by emerging blockchain governance systems, which include voters with different weights and a massive number of community generated proposals. In such scenarios, it is natural to expect that voters will have incomplete preferences, as they may only be able to evaluate or be confident about a very small proportion of the alternatives. As a result, the election outcome may be significantly affected, leading to suboptimal decisions. 
Our central inquiry revolves around whether delegation of ballots to proxies possessing greater expertise or a more comprehensive understanding of the voters' preferences can lead to outcomes with higher legitimacy and enhanced voters' satisfaction in elections where voters submit incomplete preferences. To explore its aspects, we introduce the following model: potential proxies advertise their ballots over multiple issues, and each voter either delegates to a seemingly attractive proxy or casts a ballot directly.
We identify necessary and sufficient conditions that could lead to a socially better outcome by leveraging the participation of proxies. 
We accompany our theoretical findings with experiments on instances derived from real datasets.
Overall, our results enhance the understanding of the power of delegation towards improving election outcomes.
}

\keywords{Computational Social Choice; Proxy Voting; Incomplete Preferences}

\maketitle

\section{Introduction}
\label{sec:intro}
Broadly speaking, an election refers to a voting system in which a set of participants express their preferences over a set of possible issues or outcomes, and those are aggregated into a collective decision, typically with a socially desirable objective in mind. Besides their ``traditional'' applications such as parliamentary elections or referenda, elections often underpin the livelihood of modern systems such as blockchain governance \citep{kiayias2022sok,cevallos2021verifiably} or participatory budgeting \citep{cabannes2004participatory,aziz2021participatory}. Quite often, voters are called to vote on an extremely high number of issues, rendering the accurate expression of their preferences extremely challenging.   
For instance, the Cardano blockchain uses Project Catalyst (\url{https://projectcatalyst.io}) to allocate treasury funds to community projects, and routinely receives several thousands of proposals per funding round.
Another application comes from platforms of civic participation, where the users express support on opinions or proposals \citep{halpern}.

An unfortunate consequence of these election scenarios is that the voters will inevitably have a confident opinion only for a small number of issues (henceforth \emph{proposals}), as investing enough time and effort to inform themselves about the sheer volume of proposals is clearly prohibitive. In turn, the ``direct voting'' outcome, even under the best intentions, will most likely be ineffective in capturing the desires of society, which it would, had the voters been sufficiently informed. 

A well-documented possible remedy to this situation is to allow for \emph{proxy voting} \citep{CMMMO17}, a system in which the voters \emph{delegate} their votes to \emph{proxies}. The idea is that those proxies have the time and resources to study the different proposals carefully, and vote on behalf of the voters they represent. This in fact captures voting applications more broadly, where the reason for delegation might be a reluctance to express an opinion, lack of specialized knowledge, or even limited interest. When those proxies are part of an electorate together with other voters and proxies, the resulting system is known as \emph{liquid democracy} \citep{kahng2021liquid,caragiannis2019contribution}. Liquid democracy has been scrutinized, with arguments presented in its favor \citep{becker,halpern2021defense,epistemic} and against it \citep{caragiannis2019contribution,kahng2021liquid}, and at the same time it is being employed in real-world situations \citep{paulin} including settings similar to the one studied here, like the Project Catalyst.

A takeaway message from the ongoing debate around delegative voting is that such processes might indeed be useful under the right circumstances. Extending this line of thought and motivated by the scenarios presented above, we aim to identify the \emph{potential} and \emph{limitations} of proxy voting with regard to achieving socially desirable outcomes, in settings with incomplete votes. 
More precisely, we aim to characterize what is theoretically possible with delegation, and what is impossible, even under idealized conditions.

\subsection{Our Setting and Contribution}\label{sec:our-results}

We focus on elections in which the aim is to choose one proposal to be implemented from a range of multiple proposals. We introduce a model of proxy voting, where voters have intrinsic approval preferences over all proposals, which are only partially revealed or known to the voters themselves. A set of \emph{delegation representatives (dReps)} can then advertise ballots over the proposals and the voters in turn may either delegate to a proxy, if there is sufficient agreement (i.e., over a certain \emph{agreement threshold} between the proxy's advertised ballot and the voter's revealed preferences), or vote directly. The outcome of the election is the proposal with the largest approval score, assembled by the score from the ballots of the dReps (representing voters who delegated their vote) and the voters that vote directly. The core question we pose follows:
\medskip

\noindent \emph{``Is it possible for the dReps to advertise their preferences appropriately such that the outcome of the election has an approval score that is a good approximation of the best possible approval score; which would only be achievable if all voters had full knowledge of their preferences?''} 

\subsubsection*{``Best-Case Scenario''.}

We study the aforementioned question by making the following assumptions:
\begin{enumerate}[leftmargin=*,topsep=2pt,parsep=2pt]
\item The dReps are fully informed about the preferences of the voters, i.e., they know exactly the vector of intrinsic preferences for each voter.
\item The dReps themselves do not have actual preferences and their only goal is to achieve the best possible approximation of the optimal approval score. To do so they coordinate with each other and advertise their types accordingly.
\item When there are multiple proposals with the maximum revealed score, ties are broken in favor of those with the highest intrinsic scores. 
\end{enumerate}
One should not of course expect all of these assumptions to apply in practice. We would expect the dReps to be only partially informed about the preferences of the voters (e.g., via some probabilistic model) and to exhibit some sort of rational behavior (e.g., needing to be appropriately incentivized to advertise ballots that are aligned with socially-desirable outcomes). Still, studying ``best-case scenario'' is already instructive for results in all other regimes. In particular:

\begin{itemize}[leftmargin=20pt]
\item Our \emph{negative results (inapproximability bounds)} immediately carry over to other settings as well, regardless of the choice of the dRep information model, the rationality model for the dReps, or the choice of the tie-breaking rule. In other words, we show that certain objectives are \emph{impossible}, even when a set of fully-informed dReps coordinate to achieve the best outcome, hence they are certainly impossible for any other meaningful setting.
\item Our \emph{positive results (approximation guarantees)} establish the limits of the aforementioned impossibilities: if something is not deemed impossible by our bounds, it should be the starting point of investigations for an information/rationality model chosen for the application at hand. Clearly, if our upper bounds establish that a certain number of dReps suffices to achieve a certain approximation in the ``best-case scenario'' setting, one should expect a slightly larger number of dReps to be needed in practice.
\end{itemize}

\subsubsection{Our Results.}
We firstly present a strong impossibility, namely that for any agreement threshold higher than 50\%, the best achievable approximation ratio is linear in the number of voters. On the positive side, we show that for an appropriate \emph{coherence} notion of the instance, which captures the commonalities of the set of proposals that sets of voters are informed about, meaningful approximations are possible. For the natural case of an agreement threshold of $50\%$, we show that a single dRep is capable of achieving an approximation factor of $3$, whereas only $2$ dReps are sufficient to elect the optimal proposal. Most significantly, we present general theoretical upper and lower bounds on the achievable approximations, depending on the agreement threshold, the number of dReps, and the coherence of the instance. We complement our theoretical results with a set of experiments using the MovieLens dataset \citep{harper2015movielens}, to measure the effects of proxy voting on realistic incomplete preferences.

\subsection{Related Work}

We first comment on some works that are closer in spirit to ours. \citet{reshef} propose a model with a similar objective, but focusing on the analysis of {\it sortition}, i.e., the approximation of the welfare achieved by selecting a random small-size committee. In a related direction, \citet{CMMMO17} analyze particular delegation mechanisms, under elections with samples of voters located randomly in a metric space, according to some distribution. Our approach does not consider any randomization, neither for the voting rule nor for the preferences. 
Finally, \citet{pivato2020weighted} also consider the performance of proxy voting, focusing on understanding when the proxy-elected outcome coincides with the outcome of direct voting. Again, the model of \citet{pivato2020weighted} is randomized, where the voters delegate based on the probability of agreement to a proxy, and not based on a deterministic distance function. Moreover, no analysis of approximation guarantees is undertaken in the work of \citet{pivato2020weighted}.
Our work can be seen as one that contributes to the corpus of findings in favor of proxy voting \citep{becker,halpern2021defense,epistemic}, albeit in a markedly different manner.

There is significant work within the field of computational social choice on elections with incomplete votes. One stream has focused on the identification of possible and necessary winners by exploring potential completions of incomplete profiles; see the work of \citet{lang} for an overview. Recent work has concentrated on the computational complexity of winner determination under various voting rules within the framework of incomplete information \citep{Imber, zhou, baumeister}. Another direction has studied the complexity of centralized interventions to reduce uncertainty \citep{A-HB+22} (e.g., by educating a selected set of voters or computing delegations via a centralized algorithm). 
Furthermore, there has been an exploration of the effect of minimizing the amount of information communicated \citep{kalech, ayadi} as well as of the interplay between voters' limited energy and social welfare \citep{terzop2}.  Considerable attention has been devoted to the exploration of efficient extensions of incomplete profiles to complete ones that satisfy desirable properties \citep{terzop1, lackner, elkind}. A conceptually related area focuses on distortion in voting, investigating the implications of applying rules to preferences that are less refined than the voters' intrinsic preferences \citep{anshelevich2021distortion}.
Beyond voting scenarios, similar solution concepts have been explored in domains such as fair division \citep{bouveret}, hedonic and non-cooperative games \citep{kerkmann,brill}.

\section{Election Framework and Definitions}\label{sec:prelims}
In the current section we formally describe the main attributes of the election setting we study. \medskip 

 \noindent\textbf{Proposals and Voters.}
Let $C = \{1, 2, \ldots, m\}$ be a set of candidate proposals, where for each proposal there are exactly two options: to be elected or not. Let also $N= \{1,2,\ldots,n\}$ be a set of voters responsible for determining the elected proposal. Each voter $i\in N$ is associated with \emph{approval preferences} $v_i \in \{0,1\}^m$ over the set of proposals; we refer those as true or \emph{intrinsic} preferences. Here, $1$ and $0$ are interpreted as ``accept'' (or ``support'') and ``reject'' (or ``oppose'') a proposal, respectively.  

Crucial to our model is the fact that voters do not actually know their entire intrinsic preference vector, but only a subset of it; this could be due to the fact that they have put additional effort into researching only certain proposals to verify if they indeed support them or not, but not necessarily all of them. Formally, we will say that each voter $i$ has \emph{revealed} preferences $\hat{v}_i \in \{0, 1, \bot\}^m$, where $\bot$ denotes that the voter does not have an opinion on the corresponding proposal. 
As such, we have the following relations between $v_i$ and $\hat{v}_i$:
\begin{linenomath*}
\[\forall j\in \I \;:\; \left(\hat{v}_i(j) = v_i(j)\right) \vee \left(\hat{v}_i(j) = \bot \right).\]
\end{linenomath*} 
The collection of proposals for which a voter $i$ has developed an opinion is referred to as their \emph{revealed set}, denoted by
$R_i:= \{j \in \I: \hat{v}_i(j)=v_i(j)\}$. 
Let $m_i  := |R_i|$. Each voter $i$ also has an integer weight $w_i$; in our work, it is without loss of generality to assume that $w_i =1$ for all  $i \in N$, as we can simply make $w_i$ copies of voter $i$ (with the same preferences), and all of our results go through verbatim. \medskip

 \noindent\textbf{dReps.}
A \emph{delegation representative (dRep)} is a ``special'' voter whose aim is to attract as many voters as possible to delegate their votes to her, and then participates in the election with the combined weight of those voters.
In contrast with some of the literature, and consistently with the ``best-case scenario'' motivation (see \Cref{sec:our-results}), we view dReps as unweighted agents, devoid of personal preferences over the proposals, with the responsibility of facilitating the election of a proposal that attains substantial support from the voters.

For any proposal $j\in \I$, every dRep has an \emph{advertised}, ``intended'' vote (or type), $t(j) \in \{0, 1\}$, which is visible by the voters. We assume here that dReps present votes for all proposals.\footnote{\ We could also allow dReps to abstain in some proposals, and this would not make any difference in our setting.} We will sometimes abuse the notation and refer by $t$ both to the type vector of a dRep as well as to the dRep itself.
The distance between a voter $i$ and a dRep of type $t$ is calculated using the Hamming distance function and is dependent on the revealed preferences of $i$ and the advertised type of $t$ in proposals that are revealed to $i$. Formally, let $t_{|i}$ be the projection of the type $t$ to the proposals that belong to $R_i$. Then we define the distance between a voter $i$ and a dRep $t$ as $d(i,t):=H(\hat{v}_i,t_{|i})$.\medskip

 \noindent\textbf{Agreement Threshold.}
For a voter to delegate their vote, they have to agree with the dRep in a certain number of proposals. This is captured by a threshold bound, any agreement above which results in delegation. To make this formal, we will assume that voter $i$ delegates their vote to a dRep when their distance to the dRep's type, taking into account only the voter's revealed preferences, is at most $\floor{\frac{m_i-k_i}{2}}$, where $k_i$ is a parameter quantifying the reluctance of voter $i$ to entrust their voting power to a proxy.  
Obviously $k_i\leq m_i$, for every $i$, and we will mainly focus on scenarios in which all voters have the same parameter, thus $k=k_i$, for every voter $i$. 
For example, when $k=0$, a voter delegates their vote if they agree with the dRep in at least half of the proposals in their revealed set; we will refer to this case as \emph{majority agreement} (see also \citet{fritsch2022price,constantinescu} for a use of a very similar threshold in a difference context).  
Given a dRep of type $t$, we say that $t$ \emph{attracts} a set of voters $A(t):=\{i \in N: d(i,t)\leq \floor{\frac{m_i-k_i}{2}}\}$. Additionally we define $A(D)$, for a set of dReps $D$ as 
\begin{linenomath*}
\[A(D):=\{i \in N: \exists \ t\in D \text{ s.t. } d(i,t)\leq \lfloor{\frac{m_i-k_i}{2}\rfloor}\}.\] 
\end{linenomath*}
If a voter is attracted by multiple dReps, we assume they
delegate to any of them arbitrarily; this choice makes our positive results stronger, whereas, notably, our negative results work for any choice (e.g., even for the more intuitive choice of the closest, in terms of Hamming distance, accepted dRep). \medskip  

 \noindent\textbf{Preference Profiles.}
Let $V=(v_i)_{i\in N}$ and $\hat{V}=(\hat{v}_i)_{i\in N}$. We call \emph{intrinsic preference profile} $P=(N,\I,V)$ a voting profile that contains the intrinsic preferences of the voters in $N$ on proposals from $\I$. Similarly, we call \emph{revealed preference profile} $\hat{P}=(N,\I,\hat{V})$ the voting profile that contains their revealed preferences.
Finally, $\hat{P}_D=(N,\I,\hat{V}\cup D)$ refers to the preference profile on proposals from $\I$, that contains the revealed preferences of the voters in $N$ as well as the advertised types of the dReps in $D$. \medskip

\noindent\textbf{Approval Voting Winners.}
The winner of the election is the proposal with the highest (weighted) approval score. Formally, let $sc(j)$ denote the score of a proposal $j\in \I$ in the profile $P$, i.e., the total weight of the voters $i \in N$ such that $v_i(j)=1$. A proposal $j\in \I$ is the winning proposal in the profile $P$ if $sc(j)\geq sc(j'), \forall j'\in \I$.
Similarly, we define $\hat{sc}(j)$ and $\hat{sc}_{D}(j)$ to be the score of a proposal $j\in \I$ in the profile $\hat{P}$ and $\hat{P}_D$, respectively. Note that $\hat{sc}(j)$ represents the score that proposal $j$ would attain if all voters were to vote directly and $\hat{sc}_{D}(j)$ comprises the scores of the dReps (whose weight is the total weight of the voters they have attracted) and the scores of the voters that have not delegated their votes to any dRep, i.e., that are voting directly.  
Let $\win(P):=\argmax\{sc(j),j\in \I\}$ and $\opt(P):=sc(\win(P))$. The same notions can be extended to profiles $\hat{P}$ and $\hat{P}_D$. \medskip

\noindent \emph{Tie-breaking for the winner:} We assume that $\argmax\{\cdot\}$ returns a single winning proposal rather than a winning set, according to some tie-breaking rule. Consistently with our discussion on the ``best-case scenario'' (see \Cref{sec:our-results}), we assume that the tie-breaking is always in favor of the proposal with the maximum intrinsic score. 
\medskip

\noindent Our goal is to select a set of $\lambda$ dReps that will collectively (by participating in the election and representing voters according to the submitted thresholds) ensure that a proposal of high intrinsic approval score will be elected, as formally presented in the definition of the \problem\ problem, that follows.

\begin{table}[ht]
	\centering
	\begin{tabular}{p{12cm}}  
		\toprule
\multicolumn{1}{c}{
\normalsize{\problem$(P,\hat{P},k, \lambda)$}}\\
		\midrule
\normalsize{\textbf{Input:} An intrinsic voting profile $P$ and a revealed voting profile $\hat{P}$ on
a set $\I$ of $m$ proposals and a set $N$ of $n$ weighted voters; the voters' true (resp.~revealed) preferences $V$ (resp.~$\hat{V}$); 
a parameter $k$, so that a voter $i$ is attracted by a dRep with type $t$ if $d(i,t) \le \floor{\frac{m_i-k}{2}}$; an upper bound $\lambda$ on the number of dReps.} \vspace{4pt}\\

\normalsize{\textbf{Output:} Specify type vectors for all dReps in $D$, with $|D|\leq \lambda$, 
such that $\win(P)=\win(\hat{P}_D)$.}
\\
		\bottomrule
	\end{tabular}
\end{table}

The performance of a suggested set of dReps is measured in terms of how well the intrinsic score of the winning proposal under their presence approximates the highest intrinsic approval score. Formally:

\begin{definition}
\label{def:apx}
    Let $\rho \ge 1$. We say that a set of dReps $D$ achieves a $\rho$-approximation if $sc(\win(\hat{P}_D))\geq \frac{1}{\rho}\, sc(\win(P)).$
\end{definition}

One might be inclined to believe that attracting a sufficiently large set of voters is enough to achieve a significantly low approximation ratio guarantee.
The following proposition establishes that this is not the case. It demonstrates that the attraction part is only one component towards solving \problem, and, therefore, achieving good approximations requires further insights.

\begin{restatable}{proposition}{attractionnotenough}
\label{prop:attraction-not-enough}
It is possible for a single dRep to attract $n-1$ voters and still achieve only an $n$-approximation. 
\end{restatable}

\begin{proof}

Consider an instance with $n$ voters and four proposals. 
Voters' intrinsic preferences are presented in the table below, where the revealed preferences are given in white background.

\begin{center}
\resizebox{0.36\columnwidth}{!}{
\begin{tabular}{@{}ccccc@{}}
\toprule
                                    & \textbf{$I_1$}            & \textbf{$I_2$} & \textbf{$I_3$} & \textbf{$I_4$}         \\ \midrule
\multicolumn{1}{|c}{\textbf{voter $1$}} & 
\cellcolor[HTML]{C0C0C0}1 & 
$1$              & $0$              & \multicolumn{1}{c|}{0} \\ \midrule
\multicolumn{1}{|c}{\textbf{voters $2,3,\dots,n$}} & \cellcolor[HTML]{C0C0C0}1 & $0$              & $1$              & \multicolumn{1}{c|}{1} \\ \bottomrule
\end{tabular}}
\end{center}

It is evident that $\win(P) = I_1$, and $\opt(P)=n$. At the same time the dRep that votes in favor of all proposals attracts $n-1$ out of the $n$ voters, namely all voters but voter $1$. 
This means that $\hat{sc}_D(I_1)=|A(D)|=n-1$ (where, recall that $D$ is the set of dReps, here consisting of the single aforementioned dRep), while $\hat{sc}_D(I_2)=n$, 
since both the dRep and voter $1$ vote in favor of $I_2$.
However, $sc(I_2)=1=\frac{1}{n}\opt(P)$, and consequently, the dRep that votes in favor of all proposals, attracting $n-1$ of the voters, cannot yield an approximation factor better than $n$.
\end{proof}

Before we proceed, we highlight that if the tie-breaking rule is not in favor of
electing the optimal proposal, then the situation can be much worse; in fact, the
approximation can be infinite as demonstrated by the following example.

\begin{example}
\emph{Consider an instance with $n$ voters and $m$ proposals. Voters' intrinsic preferences are presented in the table below, where the revealed preferences are given in white background.}

 \begin{center}
\begin{tabular}{@{}cccccc@{}}
\toprule
                                        & \textbf{$I_1$}                            & \textbf{$I_2$}    & \textbf{$I_3$}    & \textbf{$\cdots$} & \textbf{$I_m$}                         \\ \midrule
\multicolumn{1}{|c}{\textbf{\emph{voter $1$}}} & \cellcolor[HTML]{C0C0C0} \emph{$1$}                 & $0$                 & \emph{$1$}                 & \textbf{$\cdots$} & \multicolumn{1}{c|}{\emph{$1$}}                 \\ \midrule
\multicolumn{1}{|c}{{$\svdots$}}  & \cellcolor[HTML]{C0C0C0}{{$\svdots$}} & {$\svdots$} & {$\svdots$} &                   & \multicolumn{1}{c|}{{$\svdots$}} \\ \midrule
\multicolumn{1}{|c}{\textbf{\emph{voter $n$}}} & \cellcolor[HTML]{C0C0C0} \emph{$1$}                & $0$                 & \emph{$1$}                 & \textbf{$\cdots$} & \multicolumn{1}{c|}{\emph{$1$}}                 \\ \bottomrule
\end{tabular}
    \end{center}
\vspace{0.5cm}
\noindent \emph{It holds that
$\win(P) = I_1$, and $\opt(P)=n$. Suppose that we select as dRep the one that advertises $1$ for every proposal under consideration. Then, $A(D)=N$. However, the presence of this dRep in $D$ implies that $\hat{sc}_D(j)=n,$ for every proposal $j$. Breaking ties in an adversarial manner, results to the election of $I_2$, albeit $sc(I_2)=0$, which leads to an infinite approximation factor.}
\end{example}

We conclude the section with the definition of an important notion for our work, that of a \emph{coherent set of voters}, i.e., sets of voters with the same revealed sets. Several of our positive results will be parameterized by properties of those sets, such as the size of the largest coherent set.  

\begin{definition}
A set of voters $N'\subseteq N$ is called coherent if  
$R_i=R_{i'}, \forall i,i' \in N'$. An instance of \problem{} is called \textit{coherent} if $N$ is coherent. 
\end{definition}

Importantly, given an instance of \problem, it is computationally easy
to verify if it is coherent, or to find the largest coherent set of voters, via a simple algorithm.

\begin{lemma}
\label{alg_for_coherent}
    Given an instance of \problem, we can in polynomial time identify the largest coherent set and its size. Furthermore, we can in polynomial time decide if the instance is coherent. 
\end{lemma}
    \begin{proof}
The first part of the statement obviously implies the second, as it simply evolves checking whether the size of the largest coherent set is $n$. 
        Now, given the intrinsic preferences $P$ of an instance of \problem\ we can find the largest coherent set as follows. For every voter $i$, find a set of voters that can form a coherent set together with $i$, namely $\{i'\in N: R_i\equiv R_{i'}\}$. Among the sets obtained this way, select one of maximal cardinality. The described procedure runs in time $O(n^2m).$
    \end{proof}

\section{Theoretical Findings}\label{sec:theory}

We start our investigation with the case of a single dRep ($\lambda=1$). Our main result here is rather negative, namely that no matter how the dRep chooses their vote, the approximation ratio cannot be better than linear in the number of voters.

\begin{restatable}{theorem}{hardnocoherent}
\label{2prop:hard_no_coherent_noproof}
For a single dRep and any $k>0$, the approximation ratio of \problem\ is $\Omega(n)$.
\end{restatable}
\begin{proof}
Consider an instance with an odd number of $m>3$ proposals and $n=m-1$ voters, where $k_i >0, \forall i\in [n]$, such that:
\begin{itemize}[leftmargin=*]
\item[-] For every voter $i\in [n]$, their preferences with respect to proposal $m$ are as follows: $v_i(m)=1$ and $\hat{v}_i(m)=\bot$. \medskip

\item[-] The remaining $m-1$ proposals are partitioned in $\frac{m-1}{2}$ pairs, say $\{1,2\}, \{3, 4\}, \ldots, \{m-2, m-1\}$ and for each one of these $\frac{m-1}{2}$ pairs of proposals, say $\{j, j+1\}$, there are two distinct voters, namely voters $j$ and $j+1$, where voter $j$ votes in favor of both proposals $j$ and $j+1$ whereas voter $j+1$ votes in favor of proposal $j$ but against proposal $j+1$; for every other proposal $j'$, the preferences of these voters are $v_i(j')=0$ and $\hat{v}_i(j')=\bot$, where $i\in \{j,j+1\}$.
\end{itemize}

Say that $P$ and $\hat{P}$ are the intrinsic and revealed profiles of the created instance, respectively. Clearly, $\win(P)=m$ and $\opt(P)=n$. We claim that a single dRep, called $t$, regardless of their advertised type, will contribute to electing a proposal $i$ that satisfies $sc(i)\leq 2$, leading to an inapproximability of $\frac{n}{2}$. Towards this, first notice that $t$ cannot attract both voters of any pair. This is easy to see, as a distance of $\max\{0, \floor{\frac{m_i - k_i}{2}}\}$ for $m_i = 2$ and $k_i\ge 1$ means agreement on both revealed proposals. As a result, for any such pair of voters, $t$ will either attract zero or one voter(s), given that any two voters do not share the same preferences with respect to both their revealed proposals. 

Consider first the scenario where $A(\{t\})=\emptyset$. In this case, since $\hat{sc}_D(m)=0<\hat{sc}_D(j)$ for any proposal $j\neq m$, it holds that there exists a proposal $j'$ such that $\win(P_D)=j'$ for which $sc(j')=2$, or in other words, the direct voting will lead to the election of an outcome that is being accepted by exactly two voters. 
On the other hand, if $t$ attracts at least one voter, say $i$ where $i$ is odd (resp.~$i$ is even), then $t$ must have voted in favor of at least one proposal apart from proposal $m$, namely for proposal $i$ (resp.~for proposal $i-1$), or else $i$ would not fully agree with $t$. But then, voter $i+1$ (resp.~$i-1$), who is not attracted by $t$, is also voting in favor of proposal $i$ (resp.~for proposal $i-1$). This results to a proposal $i$ such that $\hat{sc}_D(i)=\hat{sc}_D(m)+1$. Therefore, the winning proposal is not $m$ but a proposal $i$ for which $sc(i)\leq2$.
\end{proof}

\Cref{2prop:hard_no_coherent_noproof} should be interpreted as a very strong impossibility result since it holds even in the ``best-case scenario'' (see the discussion in \Cref{sec:our-results}).
A natural follow-up question is whether some meaningful domain restriction can circumvent this impossibility. For this, we will appeal to the notion of coherent sets of voters, defined in \Cref{sec:prelims} and we will show a bounded approximation guarantee that degrades smoothly as the size of the largest coherent set grows and as the agreement threshold decreases.

\begin{restatable}{theorem}{onecoherentthm}
\label{2thm:1drep-coh}
For a single dRep, \problem{} admits an approximation ratio of $\min\Big\{n, \frac{3n(k+2)}{2|S|}\Big\}$, where $S$ is the largest coherent set of voters in the instance.
\end{restatable}

\begin{proof}

We denote by $\dRepzero$ and $\dRepone$ the delegation representatives whose advertised types with respect to a proposal $j\in \I$, are as follows:
\begin{linenomath*}
 \[
 \dRepzero(j)=\begin{cases}
     1, \text{\,if } j=\win(P),\\
     0, \text{\,otherwise.}
 \end{cases} \!\!\!
 \dRepone(j)=1,\, \forall j\in \I.
 \]
\end{linenomath*}

We begin with the following statement. It is a direct consequence of the proof of \cref{Nlem:0direct}, which we will present later on in the text, but will also come in handy for the present proof. To maintain the flow of our presentation and prevent redundant repetition of arguments, we opted to use this forward reference.

\begin{claim}
\label{cl:1/3}
    Consider a coherent instance of \problem, for any set of proposals, and any set of voters, either $1/3$ of the voters agree with $\dRepzero$ on at least half the proposals, or the winning proposal of direct voting receives a score that is at least $1/3$ of the votes.
\end{claim}

Next, we state and prove the following lemma. 

\begin{lemma}
\label{2lem:1coh}
In a coherent instance of \problem, 
        either $sc(\win(\hat{P}))\geq$ $\frac{2n}{3(k+2)},$ or $|A(\dRepzero)|\geq \frac{2n}{3(k+2)}.$ Hence, an approximation factor of $\frac{3(k+2)}{2}$ can be achieved.
\end{lemma}
\begin{proof}[Proof of \cref{2lem:1coh}] 
    Consider a coherent instance, and let $R$ be the set of proposals that are revealed to all voters. If $\win(P)\in R$, then the optimal proposal will be elected by direct voting. Therefore, we focus on the case where $\win(P)\notin R$.
    Suppose that there exists a proposal $j' \in \I$ such that $\hat{sc}(j')\geq \frac{2n}{3(k+2)}$. But then the proof follows by the fact that $sc(win(\hat{P}))\geq \hat{sc}(win(\hat{P}))\geq \hat{sc}(j')$.    
    So we can also assume that for every proposal $j\in \I$, it holds that $\hat{sc}(j)<\frac{2n}{3(k+2)}$.
To continue with the proof of the lemma, we will need the following claim.
\begin{claim}
\label{Ncl:full agreement}
    Consider a coherent instance, with $R$ being the set of proposals commonly revealed to all voters. Suppose also that $\hat{sc}(j)<\frac{2n}{3(k+2)}$ for any $j\in R$. For any $r \in [k]$ and for any set $S_r\subseteq R$ of $r$ proposals, let $Z(S_r,\dRepzero)$ be
    the set of voters that totally agree with \dRepzero{} in all proposals of $S_r$. Then $|Z(S_r,\dRepzero)| \geq \frac{(k+2-r)n}{(k+2)}.$
\end{claim}

\begin{proof}[Proof of \cref{Ncl:full agreement}]
We will prove the statement by induction on $r$. 
To prove it for $r=1$, we define $N_j:=|\{i \in N: \hat{v}_i(j)=0\}|$, for any proposal $j\in \I$. Then, by the fact that $\hat{sc}(j)<\frac{2n}{3(k+2)}<\frac{n}{(k+2)}$ it holds that $N_j\geq n- \frac{n}{(k+2)} \geq \frac{(k+1)n}{(k+2)}$, and the induction base follows. Say now that the statement holds for a fixed $r'<k$ and call $S_{r'}\subseteq R$ an arbitrary set of $r'$ proposals. We build a set $S_{r'+1}$ by adding to $S_{r'}$ an arbitrary proposal $j\in R\setminus S_{r'}$. 
 Observe that by the fact that $\hat{sc}(j)<\frac{2n}{3(k+2)}<\frac{n}{(k+2)}$, it holds that
    at least $(|Z(S_{r'},\dRepzero)|  - \frac{n}{(k+2)})$ voters from $Z(S_{r'},\dRepzero)$ are voting against $j$. Therefore,
\begin{linenomath*}  
\begin{align*}
    |Z(S_{r'+1},\dRepzero)|\geq  \frac{(k+2-r')n}{(k+2)} - \frac{n}{(k+2)}
               = \frac{(k+2-(r'+1))n}{(k+2)}. 
\end{align*}
\end{linenomath*}           
Hence, \Cref{Ncl:full agreement} indeed holds.
\end{proof}

We now continue with proving \cref{2lem:1coh}. 
We fix $r=k$ and a set $S_k\subseteq R$ of $k$ proposals, in \Cref{Ncl:full agreement}. Note that by the discussion in \Cref{sec:prelims}, we know that $k\leq m_i$ for every voter $i$, and therefore $|R|\geq k$, so that we can choose such a set $S_k$. We then apply \Cref{cl:1/3} for the (coherent) subinstance induced by the voters in $Z(S_k,\dRepzero)$ and the proposals in $\I\setminus S_k$. This implies that either at least $\frac{|Z(S_k,\dRepzero)|}{3}$ voters agree with \dRepzero{} in
at least $\ceil{\frac{|\I\setminus S_k|}{2}}$ proposals of $\I \setminus S_k$, or there is a proposal $j\in \I \setminus S_k$ with $\hat{sc}(j)\geq \frac{|Z(S_k,\dRepzero)|}{3}$. Using now \Cref{Ncl:full agreement}, we have that
$\frac{|Z(S_k,\dRepzero)|}{3} \geq \frac{2n}{3(k+2)}$, and thus the second case is infeasible, since we have assumed that $\hat{sc}(j)<\frac{2n}{3(k+2)}$, for any $j\in \I$. Coming to the first case, we note that the voters that agree with \dRepzero{} in
at least $\ceil{\frac{|\I\setminus S_k|}{2}}$ proposals of $\I \setminus S_k$, also agree with \dRepzero{} in 
all proposals of $S_k$, by definition. 
This leads to a total agreement with \dRepzero{} of at least 
\begin{linenomath*}
\[ k+\ceil*{\frac{|\I\setminus S_k|}{2}}=k+\ceil*{\frac{m-k}{2}}\geq \ceil*{\frac{m+k}{2}}.
\]
\end{linenomath*}

Therefore, for every voter $i \in Z(S_k,\dRepzero)$, it holds that $d(i,\dRepzero)\leq m-\ceil{\frac{m+k}{2}} = \floor{\frac{m-k}{2}}$. This leads to $|A(\dRepzero)| \geq |Z(S_k,\dRepzero)| \geq \frac{2n}{3(k+2)}$ which, provided that \Cref{Ncl:full agreement} holds, completes the proof of the lemma.
\end{proof}

Applying \cref{2lem:1coh} to the largest coherent set of the given instance immediately proves the statement of the theorem.
\end{proof}

An immediate but noteworthy corollary of \Cref{2thm:1drep-coh} concerns majority agreement and coherent instances.

\begin{restatable}{corollary}{zerodirectcoh}
    \label{corol:0direct-coh}
For a single dRep, \problem{} for coherent instances and majority agreement admits an approximation ratio of $3$.
\end{restatable}

The attentive reader might have observed that for majority agreement and coherent instances, the general impossibility result of \Cref{2prop:hard_no_coherent_noproof} does not apply. In that case, one might wonder what the best achievable approximation ratio is.
To partially answer this question we offer the following negative result.

\begin{restatable}{theorem}{twoinapprox}
\label{2prop:2inapprox-coh}
Let $\varepsilon>0$. For a single dRep, \problem{} does not admit a $(1.6-\varepsilon)$ approximation, even for coherent instances and majority agreement. 
\end{restatable}

\begin{proof}
\begin{figure}[]
\centering
  \includegraphics[scale=0.26]{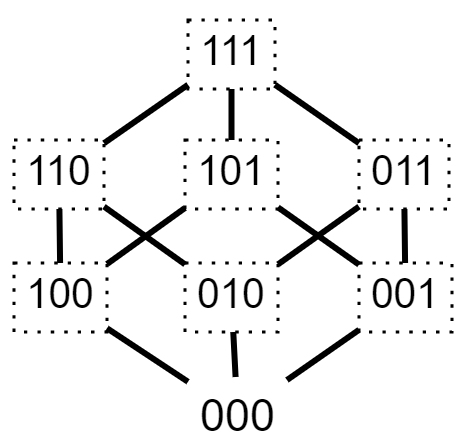}
  \caption{  
The graph illustrates the voters' preferences on the set of revealed proposals in the instance created for the proof of \cref{2prop:2inapprox-coh}. Each boxed vertex represents the existence of a voter with the corresponding ballot, with respect to proposals $c_2, c_3, c_4$. 
If a dRep advertises the ballot of a vertex $v$ in this graph, they will attract all voters whose preferences are within a distance of $1$ from $v$. } 
\label{fig:hard}
\end{figure}

Consider an instance in which $N=\{v_1,v_2,\ldots,v_8\}$ and $\I=\{c_1,c_2,c_3,c_4\}$. Suppose that $v_i(c_1)=1$ and that $\hat{v_i}(c_1)=\bot$, for every voter $i\in N$. Furthermore, there is exactly one voter whose preferences with respect to proposals $c_2,c_3,c_4$ belongs in $\{110,101,011,100,010,001\}$
and two voters that are voting for $\{111\}$. Note that in the current proof, for simplicity, we abuse the notation and use strings instead of ordered tuples to indicate voters' preferences. In this instance, $\opt(P)=sc(c_1)=8$ and $\hat{sc}(j)=5,\forall j\in \I\setminus \{c_1\}$. Therefore, direct voting cannot result in an approximation factor that is better than $\frac{8}{5}=1.6$. We will prove that for any possible choice of advertised ballot of a dRep, and if $D=\{t\}$, then $sc(\win(\hat{P}_D))\leq 5$, which again results to the claimed approximation factor. Figure \ref{fig:hard} will be helpful as an illustration of the arguments that are going to be used.

\begin{itemize}
    \item If $t=111$, then $A(D)$ equals the set of voters whose preferences belong in $\{111,110,101,011\}$. Therefore $|A(D)|=5$ and hence $\hat{sc}_D(c_1)=5$. However $c_2$ is both approved by the dRep and by a voter that doesn't belong to $A(D)$, consider, e.g., the voter who is voting for $100$, which leads to $\hat{sc}_D(c_2)=6$ and hence to a winning proposal $\win(\hat{P}_D)$ such that $\hat{sc}_D(I')\geq6$. Hence, $\win(\hat{P}_D) \neq c_1$, and, consequently, $sc(\win(\hat{P}_D))=5$. 
\item If $t=110$, then $A(D)$ equals the set of voters whose preferences are $\{111,$ $110,100,010\}$. Therefore $|A(D)|=5$ and hence $\hat{sc}_D(c_1)=5$. Using the same rationale to before, one can observe that, again, $sc(\win(\hat{P}_D))=5$. The proof is identical for $t=101$ and $t=011$.
\item If $t=100$, then $A(D)$ equals the set of voters whose preferences are $\{100,110,101\}$. Therefore $|A(D)|=3$ and hence $\hat{sc}_D(c_1)=3$. However $c_2$ is both approved by the dRep and by the voter whose preference vector is $111$, who does not belong to $A(D)$, which leads to $sc(\win(\hat{P}_D))=5$. The proof is identical for $t=010$ and $t=001$.
\item If $t=000$, then $|A(D)|=3$, but $\hat{sc}_D(c_2)=4$, which again leads to a winning proposal of $sc(\win(\hat{P}_D))=5$.\qedhere
\end{itemize}
\end{proof}

We now note that if slightly relax our best-case scenario setting (see \Cref{sec:our-results}) and consider arbitrary tie-breaking rules, then we can strengthen \Cref{2prop:2inapprox-coh} to the following impossibility result.

\begin{restatable}{proposition}{twoinapproxnew}
\label{2prop:2inapprox-coh2}
Let $\varepsilon>0$. For a single dRep, \problem{} does not admit a $(2-\varepsilon)$-approximation if ties are broken in favor of the proposal with maximum revealed score.
\end{restatable}

\begin{proof}
    The claim is straightforward if we consider an instance of two candidate proposals, namely $c_1$ and $c_2$, and two voters whose preferences are depicted in the table below. Note that only voters' preferences with respect to $c_2$ are revealed to them.
\begin{center}   
\begin{tabular}{@{}ccc@{}}
\toprule
                                    & \textbf{$c_1$}            & \textbf{$c_2$}         \\ \midrule
\multicolumn{1}{|c}{\textbf{voter $1$}} & \cellcolor[HTML]{C0C0C0}1 & \multicolumn{1}{c|}{1} \\ \midrule
\multicolumn{1}{|c}{\textbf{voter $2$}} & \cellcolor[HTML]{C0C0C0}1 & \multicolumn{1}{c|}{0} \\ \bottomrule
\end{tabular}
\end{center}
    Obviously, $\opt(P)=2$ and $\opt(\hat{P})=1$. Furthermore any dRep can attract at most one voter and hence, for any $D$ such that $|D|=1$, it holds that $\hat{sc}_D(c_1)\leq \hat{sc}_D(c_2)$. In case of ties, these are once again broken in favor of the proposal of maximum approval score in the revealed profile, i.e. in favor of $c_2$, since $\hat{sc}(c_2)>0=\hat{sc}(c_1)$. Hence, in any case, $\win(\hat{P}_D)=c_2$, and thus $sc(\win(\hat{P}_D))=1,$ which equals $\frac{1}{2}\opt(P)$.
\end{proof}

\begin{proposition}
\label{new inapprox}
    \emph{The instance used in \cref{2prop:2inapprox-coh2} can be generalized towards proving that for any acceptance threshold $k>0$, there exists an instance in which \problem{} does not admit a $(2-\epsilon)$-approximation, for any $\epsilon>0$, if $\lambda\leq 2^{k-1}$, for arbitrary tie-breaking rules.}
\end{proposition}
\begin{proof}
Say that $\lambda=2^{k-1}$ and consider a coherent instance, where $R=R_i,$ for any voter $i$, with $n=2\lambda$ voters and $m=R+1$. Given the value of $k$, we define $m$ such that $k=R-1$. The preferences of the voters are as follows:

\begin{itemize}
    \item All voters approve proposal $1$, which is the only proposal that does not belong to $R$.
    \item Only the first $\lambda$ voters approve proposal $2$. 
    \item Only the 1st and 3rd group of $\frac{\lambda}{2}$ voters approve proposal $3$.
    \item Only the 1st, 3rd, 5th, 7th group of $\frac{\lambda}{4}$ voters approve proposal $4$.\\
    $\vdots$
    \item Only voters $\{1,3,5,\ldots,n-1\}$ approve proposal $m-1$.
    \item None of the voters approve proposal $m$. 
\end{itemize}

Clearly, the optimal proposal has an approval score of $2\lambda$. Furthermore, $\ceil{\frac{R+k}{2}}=\ceil{\frac{2R-1}{2}}=R$ and hence, any dRep, can attract at most one voter, and this can be done only by advertising that voter's ballot. As a consequence, introducing any set of dReps $D$, of size no more than $\lambda$, will result in $\hat{sc}_D(1)\leq \lambda$, but at the same time $\hat{sc}_D(2)=\hat{sc}(2)=\lambda$. Breaking-ties in an adversarial manner, a proposal of intrinsic approval score $\lambda$ is being elected.
\end{proof}

To understand the power of coherence towards achieving good approximations, it is instructive to explore the limitations of the best possible dReps also on coherent instances. To this end, we provide a couple of results: 
the first generalizes \Cref{2prop:hard_no_coherent_noproof} to be parameterized by the size of the largest coherent set, and the second shows robustness to coherent instances, as long as the agreement thresholds are sufficiently high. The take-away message of those results is that coherent sets are not a panacea, and can result in meaningful approximations only under further appropriate conditions.

\begin{restatable}
{theorem}{hardnocoherentRefinement}
\label{2prop:hard_no_coherent}
Let $\varepsilon>0$. For a single dRep and any $k>0$, \problem\ does not admit an $\big(\frac{n}{|S|}-\epsilon\big)$-approximation, where $S$ is the largest coherent set of voters in the instance.
\end{restatable}

\begin{proof}
Recall that in the proof of \cref{2prop:hard_no_coherent_noproof}, we designed an instance of \problem, in which every coherent set is formed by two voters, which implied inapproximability of $\frac{n}{2}$. At what follows, we will generalize that construction to instances with coherent sets of size up to any $r > 2$. We create $r/2$ copies of each voter (more precisely, $\floor{r/2}$ copies of each voter with an odd index and $\ceil{r/2}$ copies of each voter with an even index). Clearly, the new number of voters is now $n'=n\cdot \frac{r}{2}$. 
The analysis from \cref{2prop:hard_no_coherent_noproof} directly carries over but now a proposal of score at most $r$ gets elected instead of proposal $m$ which has score $n'$. The maximum coherent set is of size $r$, and hence, the $\frac{n'}{r}$ inapproximability.
\end{proof}

\begin{restatable}
{theorem}{hardlargek}
\label{2prop:hard_large_k}
Let $\varepsilon>0$. For a single dRep and $k= m-2c(m)$, with $c(m)\in o(m)$, the approximation ratio of \problem\ is $\Omega(n)$, even for coherent instances.
\end{restatable}

\begin{proof}
We are going to construct an instance where all the voters form a single coherent set of size $n$. Moreover, $m_i = m-1$ for all $i\in [n]$.
Hence, we just write $c$ rather than 
$c(m)$, i.e., $k = m -1 -2c$. We assume that $m = (n-1)(c+1) +1$ and $n\ge 4$ (although any large enough $m$ could be used).
Further, in the created instance we have that
\begin{itemize}
\item for every voter $i\in [n]$, her preferences with respect to proposal $m$ are as follows: $v_i(m)=1$ and $\hat{v}_i(m)=\bot$, 

\item voter $i\in [n-1]$ approves proposals $(i-1)(c+1)+1, (i-1)(c+1)+2, \ldots, i(c+1)$ and disapproves everything else (except for $m$, of course, but this does not belong to $R_i$, and the preference of voter $n$ with respect to any proposal $j$ satisfies $\hat{v}_i(j)=1$.

\end{itemize}

Clearly, the optimal solution satisfies $n$ voters. Similarly to the proof of \Cref{2prop:hard_no_coherent}, we claim that a single dRep, $t$, regardless of their advertised type, will contribute  to electing a proposal approved by at most two voters. Towards this, we begin by showing that $t$ cannot attract more than one voter. 

First note that, for $i, j\in [n-1]$, we have 
 $H(\hat{v}_i, \hat{v}_j) = 2c +2.$
Also, for $i\in [n-1]$, we have 
$H(\hat{v}_i, \hat{v}_n) = m-1-(c+1) \ge
3(c+1) + 1 -1- c -1 = 2c +2,$
where we used both $m = (n-1)(c+1) +1$ and $n\ge 4$ for the inequality.
Now, suppose that $t$ attracts at least two distinct voters, say $i, j\in [n]$. Then, by the triangle inequality, we have
\begin{linenomath*}
 \begin{align*}
     2c +2\le H(\hat{v}_i, \hat{v}_j) \le d(i, t) + d(j, t) \le 2 \floor*{\frac{(m -1) - (m-1-2c)}{2}} = 2c,
 \end{align*}
\end{linenomath*}
which is a contradiction. We conclude that $t$ may attract at most one voter.

If $t$ does not attract any voters, then all voters vote directly and this leads to the election of a proposal approved by exactly two voters. 
On the other hand, if $t$ attracts one voter, say $i$, we claim that $t$ must have voted in favor of at least one proposal other than $m$. 
Indeed, if $t(j)=0$ for every proposal $j$ (except maybe for proposal $m$), then $d(i, t)$ would be at least $c+1 > \floor*{\frac{(m -1) - (m-1-2c)}{2}}$ and they would not have attracted voter $i$. As a result, a proposal of total approval $2$ in the intrinsic profile is elected, instead of proposal $m$, and the $n/2$ inapproximability follows.
\end{proof}

We conclude our discussion for the case of one dRep with a complementary result of a computational nature: a theorem that establishes that finding a dRep to attract voters in a way that ultimately elects the optimal proposal is computationally hard. Consequently, \problem\ turns out to be challenging both from the standpoint of information theory and computational complexity.

\begin{restatable}{theorem}{nphardness}\label{thm:delegation-np-hard}
  The decision variant of \problem\ is NP-hard, even for majority agreement and a single dRep.
\end{restatable}

\begin{proof}
We will establish NP-hardness for the decision version of \problem, where for some parameter $r$, we want to answer if there exists an advertised type for a dRep, so that the intrinsic score of the elected outcome is at least $r$.
We will reduce from the problem \textsc{minimax approval voting (mav)} problem, which is a known NP-hard problem in voting theory. In \textsc{mav}, we are given an instance $I$ of $m$ binary proposals and $n$ ballots where $v_i \in \{0,1\}^{m}, i\in [n]$ and we are asked for a vector $v$ for which it holds that $\max_{i\in [n]}H(v_i,v)\leq \theta$, where $H$ is the Hamming distance between two vectors of the same size. For the hardness of \textsc{mav} we refer to \citet{legrand2007some,frances}. 
We note that the NP-hardness has been established for instances with $m$ being even, and $\theta = m/2$. Given such an instance, we create an instance $I'$ of \problem{} as follows:

    \begin{itemize}
        \item We have $m' = m+3$ binary candidate proposals: $\{c_1,\ldots,c_{m}, c_{m+1}, c_{m+2}, c_{m+3}\}$, i.e., three additional proposals from $I$.
        \item We have $n$ voters corresponding to the voters of $I$, and an additional number of $n+1$ dummy voters, for a total of $2n+1$ voters.
        \item For every voter $i\in [n]$, belonging to the group of the first $n$ voters, their preferences for the first $m$ proposals in $I'$ are just as they are in $I$, and they are all revealed, so that $m_i=m$. The remaining three proposals are not visible for these voters and their intrinsic preferences are that $v_i(c_{m+1}) = 1$, $v_i(c_{m+2}) = v_i(c_{m+3}) = 0$.  
        \item For the dummy voters, none of them approve the first $m$ proposals, which are also not revealed to them. As for the last three proposals, there are exactly two dummy voters, who will be referred to as the special dummy voters, who approve all three proposals, and all three are revealed to them. All the remaining $n-1$ dummy voters approve only the proposals $c_{m+2}$ and $c_{m+3}$, which are revealed to them, whereas $c_{m+1}$ is disapproved, and also not revealed to them.
        \item We set $r = n+2$ and $\lambda=1$, i.e. we have only one dRep available. Hence we are looking for an advertised type of the dRep, so that the instrinsic score of the elected outcome is at least $n+2$.
    \end{itemize}
\begin{figure}
\centering
\resizebox{0.8\columnwidth}{!}{
\begin{tabular}{@{}ccccc@{}}
\toprule
                                                       & \textbf{m proposals}             & \textbf{$c_{m+1}$}        & \textbf{$c_{m+2}$}        & \textbf{$c_{m+3}$}                             \\ \midrule
\multicolumn{1}{|c}{\textbf{$n$ voters}}               & \cellcolor[HTML]{808080}$[\cdots]$ & \cellcolor[HTML]{C0C0C0}1 & \cellcolor[HTML]{C0C0C0}0 & \multicolumn{1}{c|}{\cellcolor[HTML]{C0C0C0}0} \\ \midrule
\multicolumn{1}{|c}{\textbf{$2$ special dummy voters}} & \cellcolor[HTML]{FFFFFF}0        & \cellcolor[HTML]{FFFFFF}1 & \cellcolor[HTML]{FFFFFF}1 & \multicolumn{1}{c|}{\cellcolor[HTML]{FFFFFF}1} \\ \midrule
\multicolumn{1}{|c}{\textbf{$n-1$ dummy voters}}       & \cellcolor[HTML]{FFFFFF}0        & \cellcolor[HTML]{C0C0C0}0 & \cellcolor[HTML]{FFFFFF}1 & \multicolumn{1}{c|}{\cellcolor[HTML]{FFFFFF}1} \\ \bottomrule
\end{tabular}%
}
    \caption{  The instance created in the proof of \Cref{thm:delegation-np-hard}. Light gray cells correspond to preferences that are not revealed to the voters and dark gray cells correspond to preferences that are derived from instance $I$ of \textsc{mav}.}
    \label{fig:hardness}
\end{figure}

An illustrative exposition of voters' ballots can be found in \cref{fig:hardness}. Before we proceed, note that the only proposal that has an intrinsic score of $n+2$ is the proposal $c_{m+1}$, while all the others have lower scores. But $c_{m+1}$ cannot be elected via only direct voting, since it is not revealed to the first $n$ voters. Hence, the question is whether there exists an advertised type for the dRep that can make $c_{m+1}$ elected.

For the forward direction, suppose that $I$ is a YES-instance of \textsc{mav}. 
Then there exists a vector $v\in \{0,1\}^{m}$ for which it holds that $\max_{i\in [n]}H(v_i,v)\leq m/2$. Consider now that the dRep advertises the type $t = (v, 1, 0, 0)$. The dRep will attract the first $n$ voters, since they agree with them in at least $m/2$ of their revealed proposals. She will not attract any of the dummy voters, all of which disagree with them on two proposals. 
Therefore, the dRep will have a weight of $n$, and since the dummy voters vote directly, the proposal $c_{m+1}$ will collect a score of n+2 and will be the winner of the election.

For the reverse direction, suppose that $I$ is a NO-instance of \textsc{mav}. Then for any possible vector $v\in \{0,1\}^{m}$, it holds that there exists at least one voter $i^*\in [n]$, such that $H(v_{i^*},v)> m/2$.
Fix now such an arbitrary vector $v\in \{0,1\}^{m}$. We will consider all possible cases for the advertised type of the dRep, $t$, and show that $c_{m+1}$ cannot get elected, i.e, $I'$ is a NO-instance. We use a natural tie-breaking rule that resolves any tie in favor of the proposal of maximum approval score in the revealed profile. Therefore, in all the cases below, any tie that involves proposals $c_{m+1},c_{m+2},c_{m+3}$ is broken against $c_{m+1}$ due to the fact that $\hat{sc}(c_{m+1})<\hat{sc}(c_{m+2})=\hat{sc}(c_{m+3})$.

\begin{itemize}
\item Suppose that $t = (v, 0, x, y)$, for any $x, y\in \{0,1\}$. It is easy to see that since the dRep does not approve $c_{m+1}$, it is not possible that this proposal wins the election.
\item Suppose that $t = (v, 1, 0, 0)$. In this case, the dRep does not attract any of the dummy voters. Hence, the proposals $c_{m+2}$ and $c_{m+3}$ receive a score of $n+1$. As for proposal $c_{m+1}$, it is crucial to note that since $I$ is a NO-instance of \textsc{mav}, $t$ can attract at most $n-1$ of the first $n$ voters. Therefore together with the two special dummy voters, the proposal $c_{m+1}$ will have a score of at most $n+1$. By the tie-breaking rule that we have assumed, this means that the winner of the election will be either $c_{m+2}$ or $c_{m+3}$.
\item Suppose that $t = (v, 1, 1, 0)$, or $t = (v, 1, 1, 1),$ or $t = (v, 1, 0, 1)$. WLOG, we analyze the former. In this case, the dRep attracts all dummy voters. Hence, the dRep has a weight of at least $n+1$ and possibly more by some of the first $n$ voters who delegate to her. Hence, the winner will either be some proposal that is approved by $v$ and also has the highest approval rate among the voters who did not delegate to dRep, or there is a tie among all the proposals approved by $t$. By the tie-breaking rule, it is not possible that $c_{m+1}$ is elected.
\end{itemize}

We have established that no matter what the dRep advertises, it is impossible that $c_{m+1}$ is elected, and hence $I'$ is a NO-instance, which concludes the proof.
\end{proof}

Importantly, we note that all of the approximation guarantees in our paper can be obtained in polynomial time.

\subsection{Multiple dReps}

We now turn our attention to the case of multiple dReps ($\lambda \geq 2)$, as this is not captured by the stark impossibility of \Cref{2prop:hard_no_coherent_noproof}. Is it perhaps possible to achieve much better approximations by using sufficiently many dReps? A reinforcing observation is that for majority agreement, 2 dReps actually suffice to elect the optimal proposal.

\begin{restatable}{theorem}{cordreps}
\label{Ncorol:02dreps}
When $\lambda= 2$, \problem\ for majority agreement can be optimally solved.\end{restatable}

\begin{proof}
We denote by \dRepallzero\ and \dRepone\ the delegation representatives whose advertised types with respect to a proposal $j\in \I$ are as follows:
 $
 \dRepallzero(j)=0,
 \dRepone(j)=1,\, \forall j\in \I
 $.
 Recall that \dRepallzero\ differs from \dRepzero\ used e.g. in the proof of \Cref{2thm:1drep-coh}. We will make use of the following lemma, the proof of which is immediate and based on the fact that any voter's ballot either contains at least as many zeros as ones or more ones than zeros.

\begin{lemma}
\label{Nobs:01new}
In an instance of \problem, with $k=0$, each voter belongs either to $A(\dRepallzero)$ or to $A(\dRepone)$ (or both).
\end{lemma}

A direct consequence of \Cref{Nobs:01new} is that if $D=\{\dRepallzero,\dRepone\}$ then $A(D)=N$. Hence, if $\lambda= 2$, we can retrieve the optimal solution. This is because each proposal will garner precisely as much support as $\win(P)$ since no delegate representative will vote in favor of a proposal without concurrently voting in favor of $\win(P)$. Therefore, $\win(P)$ will be elected, according to the tie-breaking rule.
\end{proof}

But what about instances in which voters are more discerning, indicated by larger values of $k$?
Whether good approximations with multiple delegation representatives are achievable \emph{in general} is still to be determined. We begin with coherent instances and we provide the following generalization of \Cref{Ncorol:02dreps} specifically for this case.

\begin{restatable}{theorem}{kcoherentSimple}
\label{2lem:k_coherent-simple}
When $\lambda=2^{k+1}$, \problem\ for coherent instances and any $k\geq 0$ can be optimally solved.
\end{restatable}

\begin{proof}
   Let $R$ be the set of commonly revealed proposals to the voters of the given instance.
It is without loss of generality here to assume that $\win(P) \notin R$, or otherwise, the direct voting would result in the election of $\win(P)$. 
To create a set of dReps $D$, we fix an arbitrary set $S_k\subseteq R,$ of $k$ proposals, and
    for every possible binary vector on $S_k$, i.e., for every ${\vecc{\sigma}} \in 2^{S_k}$,
    we add to $D$ exactly two dReps, namely $t_{{\vecc{\sigma}},0}$ and $t_{{\vecc{\sigma}},1},$ advertising the following, with respect to a proposal $j$: 
\begin{linenomath*}    
\[ {\small t_{{\vecc{\sigma}},0}(j)=
\begin{cases}
    \vecc{\sigma}(j), \text{\, if } j\in S_k,\\
    1, \text{\, if }j=\opt(P),\\
    0, \text{\, otherwise.}
\end{cases}\hspace{-0.1cm}
t_{{\vecc{\sigma}},1}(j)=
\begin{cases}
    \vecc{\sigma}(j), \text{\, if } j\in S_k,\\
    1, \text{\, otherwise.}
\end{cases}}
\]
\end{linenomath*}\normalsize
To prove the statement, it is sufficient, to show that $A(D)=N$, i.e., that $D$ can attract all voters from $N$. We fix an arbitrary voter $i\in N$. Definitely, there is a vector, say ${\vecc{\sigma'}}$, that defines a pair of dReps in $D$, say $t_{{\vecc{\sigma'}},0}$ and $t_{{\vecc{\sigma'}},1}$ (henceforth denoted by $t_1$ and $t_2$), such that $i$ totally agrees in all proposals of $S_k$ both with $t_1$ and $t_2$. Formally, if for a given vector $x$ and a set of proposals $Y$, we denote by $x_{|Y}$ the projection of $x$ to the proposals in $Y$, the following holds:
\begin{linenomath*}
\begin{equation}   
\label{eq:eq1}
    \max \{H(\hat{v}_i{_{|S_k}},t_{1_{|S_k}}), H(\hat{v}_i{_{|S_k}},t_{2_{|S_k}}\}=0
\end{equation}
\end{linenomath*}
\medskip
Let $R'_i := R_i \setminus S_k = R \setminus S_k$ and for $z \in \{0,1\}$ we define $R'_i(z):=|\{j \in R'_i: \hat{v}_i=z\}|.$    
Then, either $R'_i(0)\geq R'_i(1),$ or $R'_i(1)> R'_i(0)$.
    Therefore, 
    \begin{linenomath*}
    \begin{equation}
         \label{eq:eq2}
    \min \{H(\hat{v}_i{_{|R'_i}},{t_1}_{{|R'_i}}), H(\hat{v}_i{_{|R'_i}},{t_2}_{{|R'_i}}\} \leq \floor*{\frac{|R'_i|}{2}}.
    \end{equation}
    \end{linenomath*}
    Combining \Cref{eq:eq1,eq:eq2}, we have that voter $i$ agrees either with $t_1$ or with $t_2$, in at least
\begin{linenomath*}   
\[k+\ceil*{\frac{|R'_i|}{2}}=k+\ceil*{\frac{m_i-k}{2}} \geq \ceil*{\frac{m_i+k}{2}}\] 
\end{linenomath*}    
    proposals. Consequently, every voter will delegate to a dRep from $D$ and the optimal proposal will be elected.
\end{proof}

\Cref{2lem:k_coherent-simple} provides a bound on the sufficient number of dReps required to make sure that the optimal proposal is elected and raises a question regarding positive results (both optimal and approximate) for not-necessarily-coherent instances. 
In \Cref{2multipledreps} below we provide a generalized and more refined version of this result, that relates the achievable approximation with the required number of dReps and the parameters of the instance, but does not need to assume that the instances are coherent.

\subsection{Beyond Coherent Instances}\label{sec:generalizations}

Coherence has been proven to be very useful towards achieving meaningful approximation guarantees. At what follows, we define a more refined notion, namely a quantified version of it, which provides further insights into the structure of instances and how these affect the achievable approximations. In particular, we use the notion of $(x,\delta)$-coherent sets for sets of voters that have a common set of proposals of size $x$ in their revealed sets, as well as at most $\delta$ additional proposals.

\begin{definition}\label{def:delta-coherent}
A set of voters $N'\subseteq N$ is called $(x,\delta)$-coherent if there exists a set $X\subseteq \I$ such that for every $i\in N$ the following hold: $X\subseteq R_i$, $|X|\geq x$, and $|R_i\setminus X|\leq \delta$. \end{definition}

Using \cref{def:delta-coherent}, we generalize the result of \Cref{2thm:1drep-coh}, with an additional loss in the factor that is dependent on the type of $(x,\delta)$-coherent sets that an instance admits.

\begin{restatable}{theorem}{onedrepcohdelta}
\label{2thm:1drep-coh-delta}
For a single dRep and for any $\delta\!\!\!>\!\!\!0$, \problem\ admits an approximation ratio of $\min\left\{n,\frac{3n(k+\delta+2)}{2|S|}\right\}$, where $S$ is the largest $(k+\delta,\delta)$-coherent set in the instance.
\end{restatable}

\begin{proof}
By electing any proposal, one can straightforwardly obtain an approximation factor of $n$. For the remaining of the proof, say that $\frac{2n(k+\delta+2)}{2|S|}<n$. Additionally, for now and for ease of exposition, suppose that for every voter $i \in N$, it holds that $v_i(\win(P))=1$, $\hat{v}_i(\win(P))=\bot$. At the end of the proof we will show that this assumption can be dropped. For the studied case, we will prove an approximation ratio of $\frac{3n(k+\delta+2)}{2|S|}$.

The proof follows the same rationale as the proof of \Cref{2thm:1drep-coh}, which was based on applying \Cref{2lem:1coh} to the largest coherent set of the instance. Similarly, to prove \Cref{2thm:1drep-coh-delta}, it suffices to prove an analogous lemma, that we state below. Recall that 
\begin{linenomath*}
\[
 \dRepzero(j)=\begin{cases}
     1, \text{\,if } j=\win(P),\\
     0, \text{\,otherwise.}
 \end{cases} \!\!\!
 \dRepone(j)=1,\, \forall j\in \I.
 \]
 \end{linenomath*}

\begin{lemma}
\label{Nlem:B2}
In an instance of \problem, with $N$ being $(k+\delta, \delta)$-coherent, 
   it either holds that $sc(\win(\hat{P}))\geq \frac{2n}{3(k+\delta+2)}$ or $|A(\dRepzero)|\geq \frac{2n}{3(k+\delta+2)}.$
    \end{lemma}

\begin{proof}
Suppose that there exists a proposal $j' \in \I$ such that $\hat{sc}(j')\geq \frac{2n}{3(k+\delta + 2)}$. Then the proof follows by the fact that $sc(win(\hat{P}))\geq \hat{sc}(win(\hat{P}))\geq \hat{sc}(j')$. So we can assume that for every proposal $j\in \I$, it holds that $\hat{sc}(j)<\frac{2n}{3(k+\delta + 2)}$.
Next, we state a claim which is a generalization of \Cref{Ncl:full agreement}. Its proof is almost identical to the proof of \Cref{Ncl:full agreement}.
 
\begin{claim}
\label{Ncl:full agreementdelta}
    Consider an instance where the set of voters $N$ forms a $(k+\delta, \delta)$-coherent set, for some $\delta\geq 0$. Suppose also that $\hat{sc}(j)<\frac{2n}{3(k+\delta+2)},$ for any $j\in \I$. Let $X$ be the set of commonly revealed proposals to all voters, which by definition satisfies $|X|\geq k+\delta$. For any set $S_r\subseteq X$ of $r$ proposals, with $r\in [k+\delta]$, let $Z(S_r,\dRepzero)$ be as defined in \cref{Ncl:full agreement}. Then $|Z(S_r,\dRepzero)| \geq \frac{(k+2+\delta-r)n}{(k+\delta+2)}$.
\end{claim}

We proceed with proving \Cref{Nlem:B2}, using \Cref{Ncl:full agreementdelta}. Fix $r=k+\delta$ in \Cref{Ncl:full agreementdelta}, and let $S_r\subseteq X$ be a set of $k+\delta$ proposals.
Let also $\I'=X \setminus S_r$.
We examine the number of voters who will eventually delegate to \dRepzero.

We observe that every proposal of $\I'$ belongs to the revealed set of any voter of $Z(S_r,\dRepzero)$. Therefore, we can apply \Cref{corol:0direct-coh} for $Z(S_r,\dRepzero)$ and $\I'$.
This implies that either there exists a proposal in $\I'$ approved by $\frac{Z(S_r,\dRepzero)}{3}$ voters or at least $\frac{Z(S_r,\dRepzero)}{3}$ voters agree with \dRepzero{} on at least $\ceil{\frac{|\I'|}{2}}$ proposals.
By \Cref{Ncl:full agreementdelta}, we know that $|Z(S_r,\dRepzero)|\geq \frac{2n}{k+2+\delta}$.
Therefore, by the assumption we have made, it is not possible to have a proposal with a score of $\frac{|Z(S_r,\dRepzero)|}{3}$, hence, at least $\frac{2\cdot |Z(S_r,\dRepzero)|}{3(k+\delta+2)}$ voters of $Z(S_r,\dRepzero)$ agree with \dRepzero{} in at least $\ceil{\frac{|\I'|}{2}}$ proposals from $C'$. But, by the definition of $Z(S_r,\dRepzero)$, they also agree on the $k+\delta$ proposals of $S_r$. In total, their agreement with \dRepzero{} is at least:    
\begin{linenomath*}
\begin{equation*}
        \label{eq:eq4} k+\delta+\ceil*{\frac{|\I'|}{2}} \geq 
    k+\delta+\ceil*{\frac{m_i-k-\delta-\delta}{2}}  \geq \ceil*{\frac{m_i+k}{2}},
\end{equation*}
\end{linenomath*}  
where the first inequality holds due to the fact that $|C'| = |X| -(k+\delta)$, and $m_i\leq |X|+\delta$.
    As a consequence all these voters, i.e., at least $\frac{2n}{3(k+\delta+2)}$ in number, will delegate to \dRepzero.
\end{proof}

Hence, \Cref{Nlem:B2} holds and \Cref{2thm:1drep-coh-delta} follows, under the assumption that all voters secretly approve the optimal proposal. Suppose that in a given instance the assumption doesn't hold. Focusing on the set of voters, called $N'$, that satisfy the assumption and applying the described procedure, one can prove that either there is a proposal of score at least $\frac{2|N'|}{3(k+\delta+2)}$ or \dRepzero\ attracts at least $\frac{2|N'|}{3(k+\delta+2)}$ voters. Given that $sc(\win(P))=|N'|$, an approximation ratio of $\frac{3(k+\delta+2)}{2}$ is implied. 

To conclude the proof, we note that the desired approximation ratio holds true by focusing on the set $S$, just like in \cref{2thm:1drep-coh}.
\end{proof}

The main result of the subsection is a relaxation of \Cref{2lem:k_coherent-simple}. Unlike \Cref{2lem:k_coherent-simple}, the theorem does not require any structural assumptions, and relates the approximation with the number of dReps, the threshold bound and the structure of the instance in terms of approximate coherence.

\begin{restatable}{theorem}{multipledreps}
\label{2multipledreps}
When $\lambda=\min\left\{n,\zeta 2^{k+1}\right\}$, \problem\ admits an approximation ratio of $\frac{\gamma}{3\zeta}$, where $\gamma$ is the minimum number of $(k,m-k)$-coherent sets that can form a partition of $N$, and $\zeta\leq\gamma$ with $\zeta\in \mathbb{N}$. 
\end{restatable}
\begin{proof}
For ease of exposition we assume that for every voter $i \in N$, it holds that $v_i(\win(P))=1$, $\hat{v}_i(\win(P))=\bot$. We note that this assumption is without any loss in the approximation factor, likewise in the proof of \cref{2thm:1drep-coh-delta} in which we proved that focusing only on the set of voters that are secretly in favor of the optimal proposal does not affect the approximation ratio.

        Trivially, if we set $D$ to consist of one dRep $t_i$ for every voter $v_i$ such that 
\begin{linenomath*}
        \[t_i(j)=
        \begin{cases}
            1\,, \text{if } j= \win(P) \\
            v_i(j)\,, \text{otherwise}
        \end{cases}
        \]
\end{linenomath*}
    then $A(D)=N$ and hence $\win(\hat{P}_D)=\win(P)$. Therefore, with $n$ dReps we can retrieve the optimal solution and, as a consequence, the claimed approximation factor holds. 
    We will now proceed with proving that whenever $\zeta 2^{k+1}<n$, we can use a set of dReps $D$, where $|D|\leq \zeta 2^{k+1}$, to elect a proposal $j$ such that $sc(j)\geq \frac{3\zeta}{\gamma}\opt(P).$ 

We proceed with the following lemma, the proof of which is analogous to the proof of \Cref{2lem:k_coherent-simple}.

\begin{lemma}
\label{2lem:k_dcoherent}
In an instance of \problem, in which $N$ can be partitioned in at most $\gamma~ (k,m-k)$-coherent sets, \problem{} can be optimally solved if $\lambda= \gamma 2^{k+1}$. 
\end{lemma}

Using $2^{k+1}$ dReps for any of the sets described in the statement of the theorem, as indicated by  \Cref{2lem:k_dcoherent}, we can create a set of dReps $D$ such that $|D|\leq \gamma 2^{k+1}$ and $A(D)=N$. This proves the statement of the theorem for $\zeta = \gamma$. Fix now a number $\zeta < \gamma$. Among the $\gamma$ sets produced by \Cref{2lem:k_dcoherent}, let us focus on the $\zeta$ sets of largest size. Let also $N'\subseteq N$ be the voters in these sets. Then we can create a set of dReps $D'\subseteq D$, such that $|D'|\leq \zeta 2^{k+1}$ and $A(D')\supseteq N'$. But then, $|A(D')|\geq |N'|$. 

Although all dReps of $D'$ vote in favor of $\win(P)$ and this can indeed be the winning proposal of profile $P_{D'}$, it may also be the case that $\win(P_{D'})\neq \win(P)$. In fact, if some dReps from $D'$ that are voting in favor of $\win(P_{D'})$, attract voters that vote against $\win(P_{D'})$, the intrinsic score of $\win(P_{D'})$ may significantly differ from $sc(\win(P_{D'}))$. To avoid this behaviour we suggest to create a new set of dReps, say $D''$, by deleting from $D'$ all dReps of the form $t_{\sigma,1}$, i.e., all dReps that are voting in favor of proposals that do not belong in $S_k$, as defined in the proof of \cref{2lem:k_coherent-simple}. Interestingly $|D''|=\frac{|D'|}{2}$. We will now compute $|A(D'')|$. 

Suppose first that the number of voters attracted by dReps in $D'\setminus D''$ is at most $\frac{2|N'|}{3}$. Then, $|A(D'')|\geq |N'|-\frac{2|N'|}{3}=\frac{|N'|}{3}$. 

Before continuing we define $\alpha:=|\cup_{i \in \tilde{N}}R_i|$ and $\beta := \min\{|R_i|, i\in \tilde{N}\}$, where $\tilde{N}$ is any of the $\gamma$ (k,m-k)-coherent sets of the voters' partition, as described in the statement of the Theorem.

We call $N_1$ the number of voters, attracted by dReps in $D'\setminus D''$ and say that $N_1\geq \frac{2|N'|}{3}$. But then, for every voter $i\in N_1$ it holds that $R_i(1)\geq \frac{|R_i|}{2}\geq \frac{\beta}{2}$. But then, the total number of approvals in $\hat{P}$ for proposals in $\I \setminus S_k$, equals       $\sum_{i\in N_1}R_i(1)\geq |N_1|\frac{\beta}{2} \geq      
        \frac{2|N'|}{3}\frac{\beta}{2}=\frac{|N'|\beta}{3}.$ All these approvals are spread between $\alpha$ proposals. Therefore, there exists a proposal $j \in R$ for which $\hat{sc}(j)\geq \frac{|N'|\beta}{3\alpha}$.

Consequently, either there is a proposal approved by $\frac{|N'|\beta}{3\alpha}$ voters or there exists a set of dReps $D''$ that can attract $\frac{|N'|\beta}{3\alpha}$ voters without voting in favor of a proposal disapproved by a voter in $A(D'')$. Using the fact that 
$|N'|\geq \frac{\zeta}{\gamma}n\geq \frac{\zeta}{\gamma}\opt(P)$, we have an approximation ratio of $\frac{\gamma}{\zeta}\cdot\frac{\beta}{3\alpha}$. Obviously, if every of the considered $\zeta$, out of the $\gamma$ (k,m-k)-coherent sets, is coherent, then an approximation ratio of $\frac{\gamma}{3 \zeta}$ is implied.
\end{proof}

An interesting corollary of \cref{2multipledreps} (which could also be deduced from the proof of \cref{2lem:k_coherent-simple}) is the following: When aiming for an optimal solution with $2^{k+1}$ dReps, it's not a necessity for instances to be coherent; rather, the key factor is the existence of a set of $k$ proposals commonly revealed to all voters.

We conclude the discussion on generalizations with the following theorem, for the case of majority agreement ($k=0$), which generalizes \Cref{corol:0direct-coh} by relating the achievable approximation to the structure of the revealed sets, once again without the requirement of coherence.

\begin{restatable}{theorem}{zerodirect}
    \label{Nlem:0direct}
For a single dRep, \problem\ for majority agreement admits an approximation ratio of $\min\big\{n,\frac{3\alpha}{\beta}\big\}$, where $\alpha:=|\cup_{i \in N}R_i|$ and $\beta := \min\{|R_i|, i\in N\}$. 
\end{restatable}

\begin{proof}
It is apparent that any instance of \problem, is approximable within a factor of $n$, since it holds that electing arbitrarily a proposal $j\in C$ results in a winning proposal that satisfies $sc(j)\geq 1\geq \frac{1}{n}\opt(P)$. 
Therefore suppose that $\min\{n,\frac{3\alpha}{\beta}\}=\frac{3\alpha}{\beta}$. 

We proceed by showing that there exists a ballot type, the advertisement of which, results to an attraction of at least half of the voters that vote in favor of the optimal proposal. Recall that 
\begin{linenomath*}
 \[
 \dRepzero(j)=\begin{cases}
     1, \text{\,if } j=\win(P),\\
     0, \text{\,otherwise.}
 \end{cases} \!\!\!
 \dRepone(j)=1,\, \forall j\in \I.
 \]
 \end{linenomath*}
 
\noindent Note that \dRepzero\ and \dRepone\ are not the ``all-1-s'' and the ``all-0-s'' dReps, but rather they both vote in favor for the optimal proposal. Therefore, the proof of the lemma, although not involved, is not immediate, in contrast to the proof of \cref{Nobs:01new}. 

\begin{lemma}
\label{Nobs:01}
In an instance of \problem, with $k=0$, it holds that $\max\{|A(\dRepzero)|,|A(\dRepone)|\}\geq \frac{|N'|}{2},$ where $N'=\{i \in N: v_i(\win(P))=1\}.$
\end{lemma}

\begin{proof}[Proof of \cref{Nobs:01}] 
We will prove that any arbitrary voter $i\in N'$ belongs to at least one of $A(\dRepzero)$ and $A(\dRepone)$.  
We focus on the proposals $R'_i=R_i \setminus \{\win(P)\}$. Let $R'_i(0)=|\{j \in R'_i: \hat{v}_i(j)=0\}|$, and similarly, $R'_i(1)=|\{j \in R'_i: \hat{v}_i(j)=1\}|$. 
Then, trivially, it holds that either $R'_i(0)\geq R'_i(1)$ or $R'_i(1) > R'_i(0)$. If $\win(P) \in R_i$, voter $i$ agrees with both \dRepzero{}\, and \dRepone\, on $\win(P)$, whereas if $\win(P) \not\in R_i$, by the definition of the distance function, $\win(P)$ does not affect the distance of voter $i$ from any dRep. Hence in both cases, proposal $\win(P)$ does not contribute to the distance of $i$ from the two dReps, and therefore,  
\begin{linenomath*}
\[\min\{d(i,\dRepzero),d(i,\dRepone)\}\leq 
0+\left\lfloor \frac{|R_i'|}{2} \right\rfloor 
\leq 
\left\lfloor \frac{|R_i|}{2} \right\rfloor,\] 
\end{linenomath*}
where the minimum is achieved by \dRepzero{} if $R'_i(0)\geq R'_i(1)$, and by \dRepone{} otherwise. 
\end{proof}

    From \Cref{Nobs:01} it holds that either $\dRepzero$ or $\dRepone$ can attract at least half of the voters in $N'$. If this holds for $D=\{\dRepzero\}$, we are done, because either the optimal proposal wins the election, via $\dRepzero$, or another proposal wins through the voters who vote directly, in which case the intrinsic score of the winning proposal $j$ is at least $\hat{sc}_D(j)\geq A(D) \geq \frac{|N'|}{2}\geq \frac{|N'|}{3} \geq \frac{|N'|\beta}{3\alpha}= \frac{\beta}{3\alpha}\opt(P)$. 
    
    Otherwise, if \dRepone{} attracts at least half of the voters in $N'$, say that all voters from a set $N_1\subseteq N$ are being attracted by $\dRepone$. Then, for every voter $i\in N_1$, let    
    $R_i(0)=|\{j \in R_i: \hat{v}_i=0\}|$ and $R_i(1)=|\{j \in R_i: \hat{v}_i=1\}|$.    
    It should hold that $R_i(1)>R_i(0)$, due to the fact that $i\in A(\dRepone)$.
    We distinguish two cases:
    \begin{itemize}
        \item Suppose that $|N_1|\geq \frac{2|N'|}{3}$, For every voter $i \in N_1$ it holds that $R_i(1)\geq \frac{|R_i|}{2}\geq \frac{\beta}{2}$. But then, the total number of approvals in $\hat{P}$ equals       $\sum_{i\in N_1}R_i(1)\geq |N_1|\frac{\beta}{2} \geq      
        \frac{2|N'|}{3}\frac{\beta}{2}=\frac{|N'|\beta}{3}.$ All these approvals are spread between $\alpha$ proposals. Therefore, there exists a proposal $j \in R$ for which $\hat{sc}(j)\geq \frac{|N'|\beta}{3\alpha}$.
        \item Suppose now that $|N_1|<\frac{2|N'|}{3}$. But in this case, the set of voters in $N'\setminus N_1$ should be attracted by $\dRepzero$, again by \Cref{Nobs:01}. In total, these are at least $|N'|-\frac{2|N'|}{3}=\frac{|N'|}{3}$. Therefore, if $D=\{\dRepzero\}$, the delegation representative will vote with a weight of at least $\frac{|N'|}{3}$, only in favor of $\win(P)$, which would result in $\hat{sc}_D(\win(P))\geq \frac{|N'|}{3}$. Then, either $\win(P)$ will win the election achieving an optimal solution,
        or a proposal $j\in \I\setminus\{\win(P)\}$ will win. In the later case, it should hold that $\hat{sc}_D(j)\geq \hat{sc}_D(\win(P)) \geq \frac{|N'|}{3} \geq \frac{|N'|\beta}{3\alpha}$. However, $\dRepzero(j)=0$ and hence $\hat{sc}_D(j)=\hat{sc}(j)$.
    \end{itemize}
    We proved that whenever $\win(\hat{P}_D)\neq \win(P)$, there exists a proposal $j\in \I$ such that $\hat{sc}(j) \geq \frac{|N'|\beta}{3\alpha}\geq  \frac{\beta }{3\alpha}\cdot \opt(P)$. The fact that $sc(j)\geq \hat{sc}(j),$ concludes the proof.
\end{proof}

\section{Experiments}\label{sec:experiments}
We complement our theoretical results with experiments on the performance of voting with proxies on realistic data sets where voters exhibit incomplete preferences. We highlight the effect that the number of dReps, revealed set sizes and thresholds have on the total number of voters who delegate instead of voting directly, as well as the approximation compared to the optimal outcome.

The first hurdle to overcome is that, by definition, it is not possible to find datasets containing both revealed and intrinsic preferences, since voters only submit the first of the two. We use the MovieLens dataset \citep{harper2015movielens} to circumvent this issue, as there is enough contextual information to calculate plausible intrinsic preferences given the revealed ones. This set contains the reviews (with scores ranging from 0.5 to 5.0) given by 162.541 users to 62.423 movies. Of course, not every user has reviewed every movie. Each movie is characterized by a set of \emph{genres} and also has a relevance score for 1.084 different \emph{tags}, provided by the users. Using this information we calculate a plausible intrinsic user-specific ``rating'' $\texttt{Rat}_i$ for any non-reviewed movie $i$: 
\begin{linenomath*}
\begin{equation*}\texttt{Rat}_i = \frac{\sum_{j \in \mathcal{R}} \texttt{Rat}_j \cdot \texttt{sim}(i,j)}{\sum_{j \in \mathcal{R}} \texttt{sim}(i, j)},\end{equation*}
\end{linenomath*}
where $\texttt{sim}$ is a similarity metric (taking into account tag and genre relevance) and $\mathcal{R}$ is the set of movies reviewed by the user, hence $\texttt{Rat}_j$ is the rating that user gave to movie $j$. 
Note that these movie ratings are between 0.5 and 5.0, and are then converted to approval preferences by comparing them with the average rating given by that user. 

To calculate $\texttt{sim}(\cdot, \cdot)$ we use two vectors per movie $i$:
\begin{itemize}
    \item The \emph{tag} vector: $\texttt{tag}_i \in [0,1]^{1084}.$
    \item The \emph{genre} vector: $\texttt{genre}_i \in \{0, 1\}^{20}$.
\end{itemize}
The actual similarity is given by 
\[
\texttt{sim}(i, j)= 1.2^{\|\texttt{tag}_i - \texttt{tag}_j\|} \cdot \left(0.5 + \frac{\texttt{genre}_i \cdot \texttt{genre}_j}{\|\texttt{genre}_i\|\cdot \|\texttt{genre}_j\|}\right).
\]
In addition, we take a random sample of users and movies such that each user has reviewed at least 5\% of the movies and no movie has been reviewed by more than 10\% of the users. 
Specifically, we begin by sampling a large subset of 13000 users and 150 movies, and filter out the users that have reviewed fewer than 5\% of those movies, as well as the movies that have been reviewed by more than 10\% of the users. This makes the effect of delegation (and using multiple dReps in particular) more pronounced, as there are few completely uninformed users that would readily delegate and no clear pick for a best movie.
To enable meaningful comparisons, we add a movie that has not been reviewed by any user, but is approved by all in their intrinsic ratings.

\begin{figure*}[t!]
	 \centering
	\scalebox{0.96}{\input{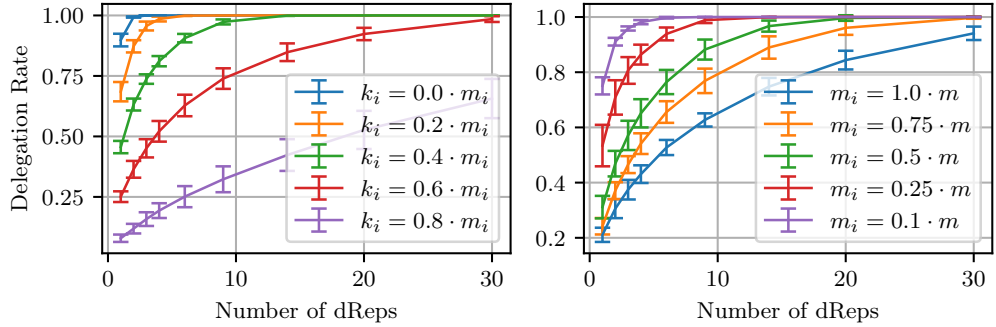}}  

 	\caption{  \normalfont{{The fraction of voters that chose to delegate as a function of the total number of dReps.}}}
   \label{fig:clustering}
\end{figure*}

\begin{figure*}[t!]
	\centering\vspace{-0.7cm}
	\scalebox{0.96}{\input{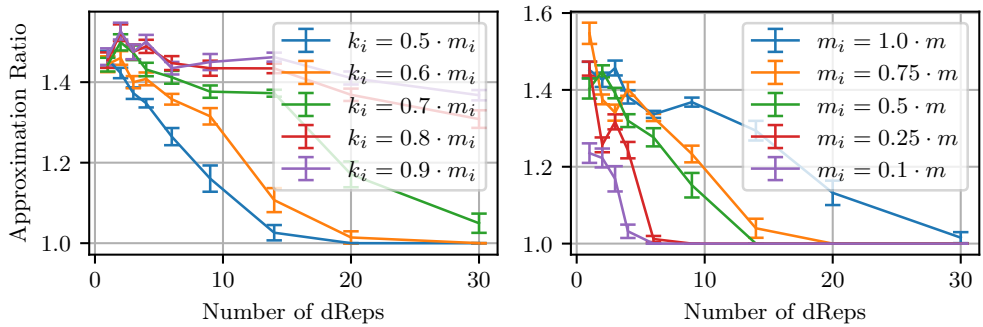}}  

	\caption{   \normalfont{{The quality of the approximation as a function of the total number of dReps. }}}
 \label{fig:approximation}
\end{figure*}

Even though in theory finding the optimal set of dReps is computationally intractable (see \Cref{thm:delegation-np-hard}), we use a greedy heuristic, which works well in practice. Specifically, we can build the type of a dRep $t$ incrementally, proposal by proposal, setting $t(j) = 1$ or $t(j) = 0$ depending on which attracts the most voters, assuming that the voters' revealed sets only include proposals up to proposal $j$.
We can repeat this procedure to create multiple dReps, removing the users that have already delegated at the end of every iteration.

We first measure the \emph{fraction} of users that opted to delegate, as a function of the number of dReps and either their $k_i$ (for fixed $m_i$) (\Cref{fig:clustering}, left) or their $m_i$ (\Cref{fig:clustering}, right). In the first case each $k_i$ ranges from $0$ up to $0.4 \, m_i$, indicating users increasing levels of user agreement required before delegating. In the second, each $m_i$ ranges from $0$ to $0.8 \, m$, capturing increasingly detailed revealed user preferences (while keeping $k_i = 0.2\, m_i$). To obtain these $m_i$s we start from the calculated intrinsic preferences (of size $m$) and we then ``hide'' some coordinates, uniformly at random, yielding the sparser revealed preferences. Our results show that it is easier to attract delegation from voters with smaller revealed sets (leading to smaller coherent sets) or with lower $k_i$'s. It turns out, that this also does indeed translate to better approximations of the optimal approval score (see \Cref{fig:approximation}; created similarly to \Cref{fig:clustering}). These results suggest that, interestingly, the case of coherent instances studied in many of our theoretical results seems to be the most challenging in practice. 

Notice that while the graphs in \Cref{fig:clustering,fig:approximation} are qualitatively similar, there are certain important differences. Specifically, the clearly diminishing returns structure observed when the only objective is to attract delegation in \Cref{fig:clustering} is not present in \Cref{fig:approximation}. This is because while it gets progressively more difficult to ``cluster'' voters, the quality of the outcomes increases in ever bigger jumps: at the very top, the difference between the best and second-best proposals (movies) will be greater than the second-best and third-best and so on, until their quality plateaus. This is also why the ``jumps'' in the approximation are steeper in some regions: a certain number of dReps does not have any effect (because even though they attract delegation, they cannot change the winner), but then suddenly this changes.

Each run is repeated 20 times for a more accurate empirical mean and standard deviation. All experiments ran on a 2021 M1 Pro MacBook with 16 gigabytes of memory and 10 high-performance cores. The code (available at \url{https://www.plazos.me/code/DelegatedVoting/}
is parallelized and requires a few hours to complete, including some initial preprocessing of the data.

\section{Discussion and Future Directions}
\label{sec:disc}
In this paper, we proposed and studied a model for proxy voting where the (less-informed) voters delegate their votes to the (fully-informed) proxies (dReps), once a certain agreement between their ballots is reached, or they vote directly otherwise. Our findings encompass essential insights into comprehending what is possible (potential) and what is not (limitations) in this setting. By identifying structural properties and other restrictions, we managed to escape the strong impossibilities that we had established.

The upper and lower bounds presented in our work are not always tight, and future work could focus on sharpening these bounds. Perhaps more interesting is the migration from the ``best-case scenario'' setting that we study in this work. This would most probably entail the following two components:

\begin{itemize}[leftmargin=*,topsep=0.2cm]
\item[-] \emph{An information model for the dReps}. It would be reasonable to assume that each delegate representative is correctly-informed about each voter $i$'s approval preference of each proposal $j$ with some probability $p_{ij}$, or that they are (perfectly or imperfectly) informed about a randomly-chosen set of proposals for each voter. 
\item[-]  \emph{A rationality model for the dReps}. Delegate representatives might not have any incentives to coordinate towards the socially-desirable outcome, and they would need to be properly incentivized to do that, e.g., via the form of payments. Additionally, dReps could even have their own preferences regarding the proposals under consideration.
\end{itemize}

One could think of many other examples or refinements of the above, and the appropriate choice of information/rationality model for the dReps would depend on the application at hand. 
Tie-breaking rules that do not necessarily favor the optimal proposal also worth studying. 
Regardless of these choices however, the results of the ``best-case scenario'' should be the starting point of any investigation into those settings. 

Besides those extensions, other directions that we see as promising routes for further research on the topic include different distance metrics, different voting rules, the multi-winner setting, elections on interdependent proposals, as well as a partial delegation setting where voters can opt to delegate only on proposals for which they have no opinion.

\newpage
\bibliography{sn-bibliography.bib}

\begin{thebibliography}{34}
\providecommand{\natexlab}[1]{#1}
\providecommand{\url}[1]{{#1}}
\providecommand{\urlprefix}{URL }
\providecommand{\doi}[1]{\url{https://doi.org/#1}}
\providecommand{\eprint}[2][]{\url{#2}}
 \bibcommenthead

\bibitem[{Alouf{-}Heffetz et~al(2022)Alouf{-}Heffetz, Bulteau, Elkind, Talmon,
  and Teh}]{A-HB+22}
Alouf{-}Heffetz S, Bulteau L, Elkind E, et~al (2022) Better collective
  decisions via uncertainty reduction. In: Proceedings of the 31st
  International Joint Conference on Artificial Intelligence (IJCAI), pp 24--30

\bibitem[{Anshelevich et~al(2021)Anshelevich, Filos-Ratsikas, Shah, and
  Voudouris}]{anshelevich2021distortion}
Anshelevich E, Filos-Ratsikas A, Shah N, et~al (2021) Distortion in social
  choice problems: The first 15 years and beyond. In: Proceedings of the 30th
  International Joint Conference on Artificial Intelligence (IJCAI) Survey
  Track., pp 4294--4301

\bibitem[{Ayadi et~al(2019)Ayadi, Amor, Lang, and Peters}]{ayadi}
Ayadi M, Amor NB, Lang J, et~al (2019) Single transferable vote: Incomplete
  knowledge and communication issues. In: Proceedings of the 18th International
  Conference on Autonomous Agents and MultiAgent Systems (AAMAS), pp 1288--1296

\bibitem[{Aziz and Shah(2021)}]{aziz2021participatory}
Aziz H, Shah N (2021) Participatory budgeting: Models and approaches. In:
  Pathways Between Social Science and Computational Social Science: Theories,
  Methods, and Interpretations. Springer, p 215--236

\bibitem[{Baumeister and Dennisen(2015)}]{baumeister}
Baumeister D, Dennisen S (2015) Voter dissatisfaction in committee elections.
  In: Proceedings of the 14th International Conference on Autonomous Agents and
  MultiAgent Systems (AAMAS), pp 1707--1708

\bibitem[{Becker et~al(2021)Becker, D’angelo, Delfaraz, and Gilbert}]{becker}
Becker R, D’angelo G, Delfaraz E, et~al (2021) Unveiling the truth in liquid
  democracy with misinformed voters. In: Proceedings of the 7th International
  Conference on Algorithmic Decision Theory (ADT), pp 132--146

\bibitem[{Bouveret et~al(2010)Bouveret, Endriss, and Lang}]{bouveret}
Bouveret S, Endriss U, Lang J (2010) Fair division under ordinal preferences:
  Computing envy-free allocations of indivisible goods. In: Proceedings of the
  19th European Conference on Artificial Intelligence (ECAI), pp 387--392

\bibitem[{Brill et~al(2016)Brill, Freeman, and Conitzer}]{brill}
Brill M, Freeman R, Conitzer V (2016) Computing possible and necessary
  equilibrium actions (and bipartisan set winners). In: Proceedings of the 30th
  AAAI Conference on Artificial Intelligence (AAAI), pp 418--424

\bibitem[{Cabannes(2004)}]{cabannes2004participatory}
Cabannes Y (2004) Participatory budgeting: a significant contribution to
  participatory democracy. Environment and urbanization 16(1):27--46

\bibitem[{Caragiannis and Micha(2019)}]{caragiannis2019contribution}
Caragiannis I, Micha E (2019) A contribution to the critique of liquid
  democracy. In: Proceedings of the 28th International Joint Conference on
  Artificial Intelligence (IJCAI), pp 116--122

\bibitem[{Cevallos and Stewart(2021)}]{cevallos2021verifiably}
Cevallos A, Stewart A (2021) A verifiably secure and proportional committee
  election rule. In: Proceedings of the 3rd ACM Conference on Advances in
  Financial Technologies, pp 29--42

\bibitem[{Cohensius et~al(2017)Cohensius, Mannor, Meir, Meirom, and
  Orda}]{CMMMO17}
Cohensius G, Mannor S, Meir R, et~al (2017) Proxy voting for better outcomes.
  In: Proceedings of the 16th Conference on Autonomous Agents and MultiAgent
  Systems (AAMAS), pp 858--866

\bibitem[{Constantinescu and Wattenhofer(2023)}]{constantinescu}
Constantinescu A, Wattenhofer R (2023) Computing the best policy that survives
  a vote. In: Proceedings of the 22nd International Conference on Autonomous
  Agents and Multiagent Systems (AAMAS), pp 2058--2066

\bibitem[{Elkind et~al(2015)Elkind, Faliszewski, Lackner, and
  Obraztsova}]{elkind}
Elkind E, Faliszewski P, Lackner M, et~al (2015) The complexity of recognizing
  incomplete single-crossing preferences. In: Proceedings of the 29th AAAI
  Conference on Artificial Intelligence (AAAI), pp 865--871

\bibitem[{Frances and Litman(1997)}]{frances}
Frances M, Litman A (1997) On covering problems of codes. Theory Computing
  Systems 30(2):113--119

\bibitem[{Fritsch and Wattenhofer(2022)}]{fritsch2022price}
Fritsch R, Wattenhofer R (2022) The price of majority support. In: Proceedings
  of the 21st International Conference on Autonomous Agents and MultiAgent
  Systems (AAMAS), pp 436--444

\bibitem[{Halpern et~al(2023{\natexlab{a}})Halpern, Halpern, Jadbabaie, Mossel,
  Procaccia, and Revel}]{halpern2021defense}
Halpern D, Halpern JY, Jadbabaie A, et~al (2023{\natexlab{a}}) In defense of
  liquid democracy. In: Proceedings of the 24th ACM Conference on Economics and
  Computation (EC), p 852

\bibitem[{Halpern et~al(2023{\natexlab{b}})Halpern, Kehne, Procaccia,
  Tucker-Foltz, and W{\"u}thrich}]{halpern}
Halpern D, Kehne G, Procaccia AD, et~al (2023{\natexlab{b}}) Representation
  with incomplete votes. In: Proceedings of the 37th AAAI Conference on
  Artificial Intelligence (AAAI), pp 5657--5664

\bibitem[{Harper and Konstan(2015)}]{harper2015movielens}
Harper FM, Konstan JA (2015) The movielens datasets: History and context. ACM
  Transactions on Interactive Intelligent Systems 5(4):1--19

\bibitem[{Imber et~al(2022)Imber, Israel, Brill, and Kimelfeld}]{Imber}
Imber A, Israel J, Brill M, et~al (2022) Approval-based committee voting under
  incomplete information. In: Proceedings of the 36th AAAI Conference on
  Artificial Intelligence (AAAI), pp 5076--5083

\bibitem[{Kahng et~al(2021)Kahng, Mackenzie, and Procaccia}]{kahng2021liquid}
Kahng A, Mackenzie S, Procaccia A (2021) Liquid democracy: An algorithmic
  perspective. Journal of Artificial Intelligence Research 70:1223--1252

\bibitem[{Kalech et~al(2011)Kalech, Kraus, Kaminka, and Goldman}]{kalech}
Kalech M, Kraus S, Kaminka GA, et~al (2011) Practical voting rules with partial
  information. In: Proceedings of the 10th International Conference on
  Autonomous Agents and MultiAgent Systems (AAMAS), pp 151--182

\bibitem[{Kerkmann and Rothe(2019)}]{kerkmann}
Kerkmann AM, Rothe J (2019) Stability in fen-hedonic games for single-player
  deviations. In: Proceedings of the 18th International Conference on
  Autonomous Agents and MultiAgent Systems (AAMAS), pp 891--899

\bibitem[{Kiayias and Lazos(2022)}]{kiayias2022sok}
Kiayias A, Lazos P (2022) Sok: Blockchain governance. In: Proceedings of the
  4th {ACM} Conference on Advances in Financial Technologies (AFT), pp 61--73

\bibitem[{Lackner(2014)}]{lackner}
Lackner M (2014) Incomplete preferences in single-peaked electorates. In:
  Proceedings of the 28th AAAI Conference on Artificial Intelligence (AAAI), pp
  742--748

\bibitem[{Lang(2020)}]{lang}
Lang J (2020) Collective decision making under incomplete knowledge: possible
  and necessary solutions. In: Proceedings of the 29th International Joint
  Conference on Artificial Intelligence (IJCAI), pp 4885--4891

\bibitem[{LeGrand et~al(2007)LeGrand, Markakis, and Mehta}]{legrand2007some}
LeGrand R, Markakis E, Mehta A (2007) Some results on approximating the minimax
  solution in approval voting. In: Proceedings of the 6th International
  Conference on Autonomous Agents and MultiAgent Systems (AAMAS), pp 1185--1187

\bibitem[{Meir et~al(2021)Meir, Sandomirskiy, and Tennenholtz}]{reshef}
Meir R, Sandomirskiy F, Tennenholtz M (2021) Representative committees of
  peers. Journal of Artificial Intelligence Research 71:401--429

\bibitem[{Paulin(2020)}]{paulin}
Paulin A (2020) An overview of ten years of liquid democracy research. In: The
  21st Annual International Conference on Digital Government Research, pp
  116--121

\bibitem[{Pivato and Soh(2020)}]{pivato2020weighted}
Pivato M, Soh A (2020) Weighted representative democracy. Journal of
  Mathematical Economics 88:52--63

\bibitem[{Revel et~al(2022)Revel, Halpern, Berinsky, and Jadbabaie}]{epistemic}
Revel M, Halpern D, Berinsky A, et~al (2022) Liquid democracy in practice: An
  empirical analysis of its epistemic performance. In: {ACM} Conference on
  Equity and Access in Algorithms, Mechanisms, and Optimization

\bibitem[{Terzopoulou(2023)}]{terzop2}
Terzopoulou Z (2023) Voting with limited energy: A study of plurality and
  {B}orda. In: Proceedings of the 22nd International Conference on Autonomous
  Agents and MultiAgent Systems (AAMAS), pp 2085--2093

\bibitem[{Terzopoulou et~al(2021)Terzopoulou, Karpov, and Obraztsova}]{terzop1}
Terzopoulou Z, Karpov A, Obraztsova S (2021) Restricted domains of dichotomous
  preferences with possibly incomplete information. In: Proceedings of the 35th
  AAAI Conference on Artificial Intelligence (AAAI), pp 5726--5733

\bibitem[{Zhou et~al(2019)Zhou, Yang, and Guo}]{zhou}
Zhou A, Yang Y, Guo J (2019) Parameterized complexity of committee elections
  with dichotomous and trichotomous votes. In: Proceedings of the 18th
  International Conference on Autonomous Agents and MultiAgent Systems (AAMAS),
  pp 503--510

\end{thebibliography}

\end{document}